\documentclass[oribibl,envcountsame,envcountsect,runningheads]{llncs}

\newif\ifignore 

\ignorefalse

\def \tnil {\langle\rangle}
\spnewtheorem{assumption}[theorem]{Assumption}{\bfseries}{\upshape}

\usepackage[english]{babel}
\usepackage{amsmath}
\usepackage{amssymb}
\usepackage{url}
\usepackage{enumerate}
\usepackage[shortlabels]{enumitem}
\usepackage{times}
\usepackage{amstext}

\usepackage{wrapfig}
\newdimen\proofrulebreadth \proofrulebreadth=.05em
\newdimen\proofdotseparation \proofdotseparation=1.25ex
\newdimen\proofrulebaseline \proofrulebaseline=2ex
\newcount\proofdotnumber \proofdotnumber=3
\let\then\relax
\def\hfi{\hskip0pt plus.0001fil}
\mathchardef\squigto="3A3B
\newif\ifinsideprooftree\insideprooftreefalse
\newif\ifonleftofproofrule\onleftofproofrulefalse
\newif\ifproofdots\proofdotsfalse
\newif\ifdoubleproof\doubleprooffalse
\let\wereinproofbit\relax
\newdimen\shortenproofleft
\newdimen\shortenproofright
\newdimen\proofbelowshift
\newbox\proofabove
\newbox\proofbelow
\newbox\proofrulename
\def\shiftproofbelow{\let\next\relax\afterassignment\setshiftproofbelow\dimen0 }
\def\shiftproofbelowneg{\def\next{\multiply\dimen0 by-1 }%
\afterassignment\setshiftproofbelow\dimen0 }
\def\setshiftproofbelow{\next\proofbelowshift=\dimen0 }
\def\setproofrulebreadth{\proofrulebreadth}

\def\prooftree{
\ifnum  \lastpenalty=1
\then   \unpenalty
\else   \onleftofproofrulefalse
\fi
\ifonleftofproofrule
\else   \ifinsideprooftree
        \then   \hskip.5em plus1fil
        \fi
\fi
\bgroup
\setbox\proofbelow=\hbox{}\setbox\proofrulename=\hbox{}%
\let\justifies\proofover\let\leadsto\proofoverdots\let\Justifies\proofoverdbl
\let\using\proofusing\let\[\prooftree
\ifinsideprooftree\let\]\endprooftree\fi
\proofdotsfalse\doubleprooffalse
\let\thickness\setproofrulebreadth
\let\shiftright\shiftproofbelow \let\shift\shiftproofbelow
\let\shiftleft\shiftproofbelowneg
\let\ifwasinsideprooftree\ifinsideprooftree
\insideprooftreetrue
\setbox\proofabove=\hbox\bgroup$\displaystyle 
\let\wereinproofbit\prooftree
\shortenproofleft=0pt \shortenproofright=0pt \proofbelowshift=0pt
\onleftofproofruletrue\penalty1
}

\def\eproofbit{
\ifx    \wereinproofbit\prooftree
\then   \ifcase \lastpenalty
        \then   \shortenproofright=0pt  
        \or     \unpenalty\hfil         
        \or     \unpenalty\unskip       
        \else   \shortenproofright=0pt  
        \fi
\fi
\global\dimen0=\shortenproofleft
\global\dimen1=\shortenproofright
\global\dimen2=\proofrulebreadth
\global\dimen3=\proofbelowshift
\global\dimen4=\proofdotseparation
\global\count255=\proofdotnumber
$\egroup  
\shortenproofleft=\dimen0
\shortenproofright=\dimen1
\proofrulebreadth=\dimen2
\proofbelowshift=\dimen3
\proofdotseparation=\dimen4
\proofdotnumber=\count255
}

\def\proofover{
\eproofbit 
\setbox\proofbelow=\hbox\bgroup 
\let\wereinproofbit\proofover
$\displaystyle
}%
\def\proofoverdbl{
\eproofbit 
\doubleprooftrue
\setbox\proofbelow=\hbox\bgroup 
\let\wereinproofbit\proofoverdbl
$\displaystyle
}%
\def\proofoverdots{
\eproofbit 
\proofdotstrue
\setbox\proofbelow=\hbox\bgroup 
\let\wereinproofbit\proofoverdots
$\displaystyle
}%
\def\proofusing{
\eproofbit 
\setbox\proofrulename=\hbox\bgroup 
\let\wereinproofbit\proofusing
\kern0.3em$
}

\def\endprooftree{
\eproofbit 
  \dimen5 =0pt
\dimen0=\wd\proofabove \advance\dimen0-\shortenproofleft
\advance\dimen0-\shortenproofright
\dimen1=.5\dimen0 \advance\dimen1-.5\wd\proofbelow
\dimen4=\dimen1
\advance\dimen1\proofbelowshift \advance\dimen4-\proofbelowshift
\ifdim  \dimen1<0pt
\then   \advance\shortenproofleft\dimen1
        \advance\dimen0-\dimen1
        \dimen1=0pt
        \ifdim  \shortenproofleft<0pt
        \then   \setbox\proofabove=\hbox{%
                        \kern-\shortenproofleft\unhbox\proofabove}%
                \shortenproofleft=0pt
        \fi
\fi
\ifdim  \dimen4<0pt
\then   \advance\shortenproofright\dimen4
        \advance\dimen0-\dimen4
        \dimen4=0pt
\fi
\ifdim  \shortenproofright<\wd\proofrulename
\then   \shortenproofright=\wd\proofrulename
\fi
\dimen2=\shortenproofleft \advance\dimen2 by\dimen1
\dimen3=\shortenproofright\advance\dimen3 by\dimen4
\ifproofdots
\then
        \dimen6=\shortenproofleft \advance\dimen6 .5\dimen0
        \setbox1=\vbox to\proofdotseparation{\vss\hbox{$\cdot$}\vss}%
        \setbox0=\hbox{%
                \advance\dimen6-.5\wd1
                \kern\dimen6
                $\vcenter to\proofdotnumber\proofdotseparation
                        {\leaders\box1\vfill}$%
                \unhbox\proofrulename}%
\else   \dimen6=\fontdimen22\the\textfont2 
        \dimen7=\dimen6
        \advance\dimen6by.5\proofrulebreadth
        \advance\dimen7by-.5\proofrulebreadth
        \setbox0=\hbox{%
                \kern\shortenproofleft
                \ifdoubleproof
                \then   \hbox to\dimen0{%
                        $\mathsurround0pt\mathord=\mkern-6mu%
                        \cleaders\hbox{$\mkern-2mu=\mkern-2mu$}\hfill
                        \mkern-6mu\mathord=$}%
                \else   \vrule height\dimen6 depth-\dimen7 width\dimen0
                \fi
                \unhbox\proofrulename}%
        \ht0=\dimen6 \dp0=-\dimen7
\fi
\let\doll\relax
\ifwasinsideprooftree
\then   \let\VBOX\vbox
\else   \ifmmode\else$\let\doll=$\fi
        \let\VBOX\vcenter
\fi
\VBOX   {\baselineskip\proofrulebaseline \lineskip.2ex
        \expandafter\lineskiplimit\ifproofdots0ex\else-0.6ex\fi
        \hbox   spread\dimen5   {\hfi\unhbox\proofabove\hfi}%
        \hbox{\box0}%
        \hbox   {\kern\dimen2 \box\proofbelow}}\doll%
\global\dimen2=\dimen2
\global\dimen3=\dimen3
\egroup 
\ifonleftofproofrule
\then   \shortenproofleft=\dimen2
\fi
\shortenproofright=\dimen3
\onleftofproofrulefalse
\ifinsideprooftree
\then   \hskip.5em plus 1fil \penalty2
\fi
}

\newcommand{\Pow}{\mathcal{P}}

\newcommand{\after}{\mathrel{\circ}}
\newcommand{\idmap}{\textrm{id}}

\newcommand{\cat}[1]{\ensuremath{\mathbf{#1}}}
\newcommand{\Cat}[1]{\ensuremath{\mathbf{#1}}}
\newcommand{\Sets}{\Cat{Sets}}

\newcommand{\Kl}{\mathcal{K}{\kern-.2ex}\ell}
\newcommand{\EM}{\mathcal{E}{\kern-.2ex}\mathcal{M}}

\newcommand{\lift}[1]{\smash{\widehat{#1}}}
\newcommand{\free}[1]{{#1^*}}

\newcommand{\klafter}{\circ}

\newcommand{\toFinal}[1]{{!_{#1}}}
\DeclareMathSymbol{\fromInit}{\mathord}{operators}{"3C}

\newcommand{\N}{\mathbb{N}} 
\newcommand{\C}{\cat{C}} 

\newcommand{\cppo}{\cat{Cppo}}
\newcommand{\J}{\mathcal{J}}

\newcommand{\state}[1]{*++[F-:<10pt>]{#1}}
\newcommand{\fstate}[1]{*++[F=:<10pt>]{#1}}

\mathchardef\ls="213C    
\mathchardef\gr="213E    

\renewcommand{\>}{\rangle}
\usepackage{eucal}
\usepackage{multirow}
\usepackage[cmtip,all]{xy}

\newcommand{\bb}[1]{[\![ #1 ]\!]}

\newcommand{\congrightarrow}{\mathrel{\stackrel{
           \raisebox{.5ex}{$\scriptstyle\cong\,$}}{
           \raisebox{0ex}[0ex][0ex]{$\rightarrow$}}}}

\newcommand{\takeout}[1]{\empty}

\def\id{\mathrm{id}}
\def\Id{\mathrm{Id}}
\def\sol#1{{#1}^\bullet}
\def\eps{\epsilon}

\def\inl{\mathsf{inl}}
\def\inr{\mathsf{inr}}
\def\refeq#1{(\ref{#1})}
\def\colim{\mathop{\textrm{colim}}}

\def\lsol#1{{#1}^\dag} 
\def\gensol#1{{#1}^\dag} 
\def\cansol#1{{#1}^\dag}
\def\altsol#1{{#1}^\ddag}
\def\quotsol#1{{#1}^{\sim}}
\def\epselim#1{{#1}\!\setminus\! \epsilon}
\def\carrier{I}
\newcommand{\circone}[1]{\mbox{\textcircled{\scriptsize #1}}}%

\newcommand{\bbq}[1]{\bb{#1}_{\sim}} 

\def\To{\Rightarrow}
\def\:{\colon}
\def\GF{F} 
\def\GFG{G} 
\def\MM{R} 
\def\quot{\xi} 
\def\quotG{\gamma} 

\usepackage{ifthen}

\newtheorem{definition_th}{Definition}[section]  

\newtheorem{algo_th}[definition_th]{Algorithm}

\newenvironment{theorem_for}[2][\empty]{\bigskip\noindent{\bf
    Theorem~\ref{#2}}\ifthenelse{\equal{#1}{\empty}}{{\bf.}}{ {\bf (#1).}}\it}{\vspace{0.5cm}}
\newenvironment{corollary_for}[2][\empty]{\bigskip\noindent{\bf
    Corollary~\ref{#2}}\ifthenelse{\equal{#1}{\empty}}{{\bf.}}{ {\bf (#1).}}\it}{\vspace{0.5cm}}
\newenvironment{proposition_for}[2][\empty]{\bigskip\noindent{\bf
    Proposition~\ref{#2}}\ifthenelse{\equal{#1}{\empty}}{{\bf.}}{ {\bf
      (#1).}}\it}{\vspace{0.5cm}}

\begin{document}

\title{How to Kill Epsilons with a Dagger}
\subtitle{A Coalgebraic Take on Systems with Algebraic Label Structure}

\author{Filippo Bonchi\inst{1} \and Stefan Milius\inst{2} \and Alexandra Silva\inst{3} \and Fabio Zanasi\inst{1} }
\institute{ENS Lyon, U. de Lyon, CNRS, INRIA, UCBL, France \and Lehrstuhl f\"ur Theoretische Informatik, Friedrich-Alexander Universit\"at Erlangen-N\"urnberg \and Institute for Computing and Information Sciences,
  Radboud University Nijmegen\thanks{\scriptsize Also affiliated to Centrum Wiskunde \& Informatica (Amsterdam, The Netherlands) and HASLab / INESC TEC, Universidade do Minho (Braga, Portugal).} }

\maketitle

\begin{abstract}
We propose an abstract framework for modeling state-based systems with internal behavior as e.g. given by silent or $\epsilon$-transitions. Our approach employs monads with a parametrized fixpoint operator $\dagger$ to give a semantics to those systems and implement a sound procedure of abstraction of the internal transitions, whose labels are seen as the unit of a free monoid. More broadly, our approach extends the standard coalgebraic framework for state-based systems by taking into account the algebraic structure of the labels of their transitions. This allows to consider a wide range of other examples, including Mazurkiewicz traces for concurrent systems.
\end{abstract}

\section{Introduction}\label{Sec:Intro}
The theory of coalgebras provides an elegant mathematical framework to express the semantics of computing devices:
the operational semantics, which is usually given as a state machine, is modeled as a coalgebra for a functor; the denotational semantics as the unique map into the final coalgebra of that functor. While  the denotational semantics is often \emph{compositional} (as, for instance, ensured by the bialgebraic approach of \cite{plotkin-semop}), it is sometimes not \emph{fully-abstract}, i.e, it discriminates systems that are equal from the point of view of an external observer. This is due to the presence of internal transitions (also called $\epsilon$-transitions) that are not observable but that are not abstracted away by the usual coalgebraic semantics using the unique homomorphism into the final coalgebra.


In this paper, we focus on the problem of giving trace semantics to systems with internal transitions.
Our approach stems from an elementary observation (pointed out in previous work, e.g. \cite{Sobocinski2012}): the labels of transitions form a monoid and the internal transitions are those labeled by the unit of the monoid. Thus, there is an \emph{algebraic structure} on the labels that needs to be taken into account when modeling the denotational semantics of those systems.
%
%
%
To illustrate this point, consider the following two non-deterministic automata (NDA).
$$\small
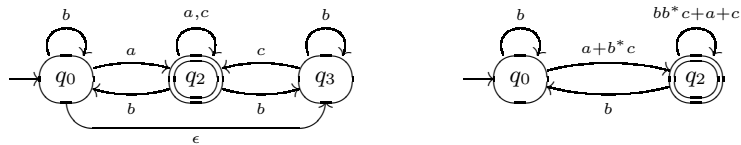
\xymatrix@C=.4cm {\ar[r]&\state{q_0}\ar@(ul,ur)^{b}
\ar@/^/[rr]^a&&\fstate{q_2}\ar@/^/[ll]^b\ar@/_/[rr]_{b}\ar@(ul,ur)^{a,c}&& \state{q_3}\ar@(ul,ur)^{b}\ar@/_/[ll]_{c}\ar@{<-} `d[l] `[llll]^\epsilon [llll]} 
\qquad\qquad
\xymatrix@C=.4cm {\ar[r]&\state{q_0}\ar@/^/[rrr]^{a+b^*c}\ar@(ul,ur)^{b}&&
&\fstate{q_2}\ar@/^/[lll]^{b}\ar@(ul,ur)^{bb^*c+a+c} }
$$
The one on the left (that we call $\mathbb{A}$) is an NDA with $\epsilon$-transitions: its transitions are labeled either by the symbols of the alphabet $A=\{a,b,c\}$ or by the empty word $\epsilon\in A^*$. The one on the right (that we call $\mathbb{B}$) has transitions labeled by languages in $\Pow(A^*)$, here represented as regular expressions.
The monoid structure on the labels is explicit on $\mathbb{B}$, while it is less evident in $\mathbb{A}$ since the set of labels $A\cup \{\epsilon \}$ does not form a monoid. However, this set can be trivially embedded into $\Pow(A^*)$ by looking at each symbols as the corresponding singleton language. For this reason each automaton with $\epsilon$-transitions, like $\mathbb{A}$, can be regarded as an automaton with transitions labeled by languages, like $\mathbb{B}$. Furthermore, we can define the semantics of NDA with $\epsilon$-transitions by defining the semantics of NDA with transitions labeled by languages: a word $w$ is accepted by a state $q$ if there is a path $\xymatrix@C=.4cm {q\ar[r]^{L_1} & \cdots  \ar[r]^{L_n} & p}$ where $p$ is a final state, and there exist a decomposition $w=w_1\cdots w_n$ such that $w_i \in L_i$.
Observe that, with this definition, $\mathbb{A}$ and $\mathbb{B}$ accept the same language: all words over $A$ that end with $a$ or $c$. In fact, $\mathbb{B}$ was obtained from $\mathbb{A}$ in a well-known process to compute the regular expression denoting the language accepted by a given automaton~\cite{Hopcroft}.

We propose to define the semantics of systems with internal transitions following the same idea as in the above example. Given some transition type (i.e.~an endofunctor) $F$, one first defines an embedding of $F$-systems with internal transitions into $F^*$-system, where $F^*$ has been derived from $F$ by making explicit the algebraic structure on the labels. Next one models the semantics of an $F$-system as the one of the corresponding $F^*$-system $e$. Naively, one could think of defining the semantics of $e$ as the unique map $\toFinal{e}$ into the final coalgebra for $F^*$. However, this approach turns out to be too fine grained, essentially because it ignores the underlying algebraic structure on the labels of $e$. The same problem can be observed in the example above: $\mathbb{B}$ and the representation of $\mathbb{A}$ as an automaton with languages as labels have different final semantics---they accept the same language only modulo the equations of monoids.

Thus we need to extend the standard coalgebraic framework by taking into account the algebraic structure on labels. To this end, we develop our theory for systems whose transition type $F^*$ has a \emph{canonical fixpoint}, i.e. its initial algebra and final coalgebra coincide. This is the case for many relevant examples, as observed in \cite{HasuoJS:07}. Our \emph{canonical fixpoint semantics} will be given as the composite $\fromInit \circ \toFinal{e}$, where $\toFinal{e}$ is a coalgebra morphism given by finality and $\fromInit$ is an algebra morphism given by initiality. 
The target of $\fromInit$ will be an algebra for $F^*$ encoding the equational theory associated with the labels of $F^*$-systems. Intuitively, $\fromInit$ being an \emph{algebra} morphism, will take the quotient of the semantics given by $\toFinal{e}$ modulo those equations. Therefore the extension provided by $\fromInit$ is the technical feature allowing us to take into account the algebraic structure on labels.

To study the properties of our canonical fixpoint semantics, it will be convenient to formulate it as an operator $e \mapsto \cansol{e}$ assigning to systems (seen as sets of equations) a certain \emph{solution}. Within the same perspective we will implement a different kind of solution $e \mapsto \altsol{e}$ turning any system $e$ with internal transitions into one $\altsol{e}$ where those have been abstracted away. By comparing the operators $e \mapsto \cansol{e}$ and $e \mapsto \altsol{e}$, we will then be able to show that such a procedure (also called \emph{$\epsilon$-elimination}) is sound with respect to the canonical fixpoint semantics.

To conclude, we will explore further the flexibility of our framework. In particular, we will model the case in which the algebraic structure of the labels is quotiented under some equations, resulting in a coarser equivalence than the one given by the canonical fixpoint semantics. As a relevant example of this phenomenon, we give the first coalgebraic account of Mazurkiewicz traces.

\paragraph{Synopsis} After recalling the necessary background in Section \ref{Sec:Trace}, we discuss our motivating examples---automata with $\epsilon$-transitions and automata on words---in Section \ref{SSec:Mot}. Section \ref{sec:Theory} is devoted to present the canonical fixpoint semantics and the sound procedure of $\epsilon$-elimination. This framework is then instantiated to the examples of Section \ref{SSec:Mot}. Finally, in Section \ref{ssec:quot} we show how a quotient of the algebra on labels induces a coarser canonical fixpoint semantics. We propose Mazurkiewicz traces as a motivating example for such a construction. A full version of this paper with all proofs and extra material can be found in \url{http://arxiv.org/abs/1402.4062}.

\section{Preliminaries}\label{Sec:Trace}
In this section we introduce the basic notions we need for our abstract framework. We assume some familiarity with
category theory. We will use boldface capitals $\cat{C}$ to denote categories, $X, Y, \ldots$ for objects and $f,g,\ldots$ for morphisms.
We use Greek letters and double arrows, e.g. $\eta \colon F \To G$, for natural transformations, monad morphisms and any kind of 2-cells. If $\mathbf{C}$ has coproducts we will denote them by $X+Y$ and use $\inl,\inr$ for the coproduct injections.

\subsection{Monads}

We recall the basics of the theory of monads, as needed here. For more information, see \textit{e.g.}~\cite{MacLane71}. A monad is a functor $T\colon\cat{C} \rightarrow \cat{C}$ together with two natural transformations, a \emph{unit} $\eta\colon \idmap_{\cat{C}} \To T$ and a \emph{multiplication} $\mu \colon T^{2} \To T$, which are required to satisfy the following equations, for every $X\in\cat{C}$: $\mu_X\after\eta_{T X} = \mu_X\after T\eta_{X} = \id$ and $T\mu_{X}\after \mu_{T X} = \mu_X\after \mu_X$.

A \emph{morphism of monads} from $(T,\eta^T, \mu^T)$ to $(S,\eta^S, \mu^S)$ is a natural transformation $\gamma \colon T \Rightarrow S$ that preserves unit and multiplication: $\gamma_X \circ \eta^T_X=\eta^S_X$ and $\gamma_X \circ \mu^T_X=\mu_X^S \circ \gamma_{SX} \circ T\gamma_X$. A \emph{quotient of monads} is a morphism of monads with epimorphic components.

\begin{example}\label{ex:mnds}
We briefly describe the examples of monads that
we use in this paper.
\begin{enumerate}
\item \label{pt:powersetmonad} Let $\cat{C}=\Sets$. The powerset monad $\Pow$ maps a set $X$ to the set $\Pow X$ of
  subsets of $X$, and a function $f\colon X \to Y$ to $\Pow  f\colon
  \Pow X \to \Pow Y$ given by direct image. The unit is given by the
  singleton set map $\eta_X(x) = \{x\}$ and multiplication by union
  $\mu_X(U) = \bigcup_{S \in U} S$.

\item \label{pt:exceptionmonad}  Let $\cat{C}$ be a category with coproducts and $E$ an object of $\cat C$. The exception monad $\mathcal E$ is defined on objects as $\mathcal E  X = E + X$
and on arrows $f\colon X\to Y$ as $\mathcal E   f = \Id_E + f$. Its unit and multiplication are given on $X \in \cat C$ respectively as $\inr_X\colon X \to E + X$ and $\nabla_{E} + \Id_X : E + E + X\to E + X$, where $\nabla_E = [\id_E,\id_E]$ is the codiagonal. When $\cat C=\Sets$, $E$ can be thought as a set of {\em exceptions} and this monad is often used to encode computations that might fail throwing an exception chosen from the set $E$.

\item \label{pt:freemonad} Let $H$ be an endofunctor on a category $\C$ such that for every
  object $X$ there exists a free $H$-algebra $\free{H}X$ on $X$
  (equivalently, an initial
  $H+X$-algebra) with the structure $\tau_X: H\free{H} X \to \free{H}X$ and
  universal morphism $\eta_X: X \to \free{H} X$. Then as proved by Barr~\cite{barr:70} (see also
  Kelly~\cite{kelly:80}) $\free{H}\colon \C \to \C$ is the functor part of a \emph{free
  monad} on $H$ with the unit given by the above $\eta_X$ and the
  multiplication given by the freeness of $\free{H}\free{H}X$: $\mu_X$ is the
  unique $H$-algebra homomorphism from $(\free{H}\free{H}X, \tau_{\free{H}X})$ to
  $(\free{H}X , \tau_X)$ such that $\mu_X \cdot \eta_{\free{H}X} = \eta_X$.
  Also notice that for a complete category every free monad
  arises in this way. Finally, for later use we fix the notation $\kappa = \tau
  \cdot H\eta\colon H \To \free{H}$ for the universal natural transformation of the free monad.
\end{enumerate}
\end{example}
Given a monad $M\colon \C \to \C$, its \emph{Kleisli category} $\Kl(M)$ has the same objects as $\cat{C}$, but morphisms $X\rightarrow Y$ in $\Kl(M)$ are morphisms $X\rightarrow MY$ in $\cat{C}$. The identity map $X\rightarrow X$ in $\Kl(M)$ is $M$'s unit $\eta_{X}\colon X\rightarrow M X$; and composition $g \klafter f$ in $\Kl(M)$ uses $M$'s multiplication: $g \klafter f = \mu \after Mg \after f$. There is a forgetful functor $\mathcal{U}\colon\Kl(T)\rightarrow \cat{C}$, sending $X$ to $T X$ and $f$ to $\mu \after T  f$. This functor has a left adjoint $\J$, given by $\J X = X$ and $\J  f = \eta \after f$. The Kleisli category $\Kl(M)$ inherits coproducts from the underlying category $\cat{C}$. More precisely, for every objects $X$ and $Y$ their coproduct $X+Y$ in $\cat C$ is also a coproduct in $\Kl(M)$ with the injections $\J\inl$ and $\J\inr$.
\subsection{Distributive laws and liftings}\label{SSec:Distributive}
The most interesting examples of the theory that we will present in Section~\ref{sec:Theory} concern coalgebras
for functors $\lift{H}\colon \Kl(M) \to \Kl(M)$ that are obtained as liftings of endofunctors $H$ on $\Sets$.
Formally, given a monad $M\colon \C \to \C$, a \emph{lifting} of $H \colon \C \to \C$ to $\Kl(M)$ is an endofunctor
 $\lift{H}\colon \Kl(M) \to \Kl(M)$ such that $\J \circ H = \lift{H} \circ \J$. The lifting of a monad $(T,\eta, \mu)$ is a monad $(\lift{T},\lift{\eta}, \lift{\mu})$ such that
$\lift{T}$ is a lifting of $T$ and $\lift{\eta}$, $\lift{\mu}$ are given on $X \in \Kl(M)$ (i.e. $X \in \Sets$) respectively as $\J(\eta_{X})$ and $\J(\mu_{X})$.

A natural way of lifting functors and monads is by mean of distributive laws.
A {\em distributive law} of a monad $(T,\eta^T,\mu^T)$ over a monad $(M,\eta^M,\mu^M)$ is a
natural transformation $\lambda \colon
TM\To MT$, that commutes
appropriately with the unit and multiplication of both monads; more precisely, the diagrams below commute:
\begin{equation*}
\vcenter{\xymatrix@R-.5pc{
TX\ar[d]_{T\eta^M_{X}}\ar@{=}[r] & TX\ar[d]^{\eta^M_{TX}}
& \quad
TM^{2}X\ar[d]_{T\mu^M_{X}}\ar[r]^-{\lambda_{MX}} &
   MTMX\ar[r]^-{M\lambda_{X}} &
   M^{2}TX\ar[d]^{\mu^M_{TX}} \\
TMX\ar[r]_-{\lambda_X} & MTX
& \quad
TMX\ar[rr]_-{\lambda_X} & & MTX\\
MX\ar[u]^{\eta^T_{MX}}\ar@{=}[r] & MX\ar[u]_{M\eta^T_{X}}
& \quad
T^{2}MX\ar[u]^{\mu^T_{MX}}\ar[r]_{T\lambda_{MX}} &
   TMTX\ar[r]_{\lambda_{TX}} &
   T^{2}MX\ar[u]_{M\mu^T_{X}}  }}
\end{equation*}
A distributive law of a {\em functor } $T$ over a {\em monad} $(M,\eta^M,\mu^M)$ is a
natural transformation $\lambda \colon
TM\To MT$ such that only the two topmost squares above commute.

The following ``folklore'' result gives an alternative description of
distributive laws in terms of liftings to Kleisli categories, see also~\cite{Johnstone75}, \cite{Mulry93} or~\cite{BalanK11}.
\begin{proposition}[\cite{Mulry93}]
\label{LiftProp}
Let $(M,\eta^M,\mu^M)$ be a monad on a category $\C$. Then the following holds:
\begin{enumerate}
\item For every endofunctor $T$ on $\C$, there is a bijective correspondence between liftings of $T$ to $\Kl(M)$ and distributive laws of $T$ over $M$.
\item For every monad $(T,\eta^T,\mu^T)$ on $\C$, there is a bijective correspondence between liftings of $(T,\eta^T,\mu^T)$ to $\Kl(M)$ and distributive laws of $T$ over $M$.
 \end{enumerate}
\end{proposition}
In what follows we shall simply write $\lift{H}$ for the
lifting of an endofunctor $H$.
\begin{proposition}[\cite{HasuoJS:07}]\label{prop:liftinginitialalgebra}
 Let $M \colon \C \to \C$ be a monad and $H \colon \C \to \C$ be a functor with a lifting
 $\lift{H} \colon \Kl(M) \to \Kl(M)$.
 If $H$ has an initial algebra $\iota \colon H\carrier \congrightarrow  \carrier$ (in $\cat{C}$),
 then $\J\iota \colon \lift{H}\carrier \to \carrier$ is an initial algebra for $\lift{H}$ (in $\Kl(M)$).
\end{proposition}
In our examples, we will often consider the free monad (Example \ref{ex:mnds}.\ref{pt:freemonad}) $\free{\lift{H}}$ generated by a lifted functor $\lift{H}$. The following result will be pivotal.
\newcommand{\propliftingfreemonad}{ Let $H \colon \C \to \C$ be a functor and $M \colon \C \to \C$ be a monad such that there is a lifting
 $\lift{H} \colon \Kl(M) \to \Kl(M)$. Then the free monad $\free{H} \colon \C \to \C$ lifts to a monad
 $\lift{\free{H}} \colon \Kl(M) \to \Kl(M)$. Moreover, $\lift{\free{H}} = \free{\lift{H}}$.}
\begin{proposition}\label{prop:liftingfreemonad}
\propliftingfreemonad
\end{proposition}

Recall from~\cite{HasuoJS:07} that for every polynomial endofunctor
$H$ on $\Sets$ there exists a canonical distributive law of $H$ over
any \emph{commutative} monad $M$ (equivalently, a canonical lifting of
$H$ to $\Kl(M)$); this result was later extended to so-called analytic
endofunctors of $\Sets$ (see~\cite{mps:09}). This can be used in our
applications since the power-set functor $\Pow$ is commutative, and so
is the exception monad $\mathcal E$ iff $E = 1$.

\subsection{$\cppo$-enriched categories}
\label{sec:cppo}

For our general theory we are going to assume that we work in a
category where the hom-sets carry a cpo structure. Recall that a
\emph{cpo} is a partially ordered set in which all $\omega$-chains
have a join. A cpo with bottom is a cpo with a least element $\bot$. A
function between cpos is called \emph{continuous} if it preserves
joins of $\omega$-chains. Cpos with bottom and continuous maps form a category
that we denote by $\cppo$.

A \emph{\cppo-enriched category} $\cat{C}$ is a category where (a)
each hom-set $\cat{C}(X,Y)$ is a cpo 
with a bottom element $\bot_{X,Y}\colon X \to Y$ and
(b) composition is continuous, that is:
$$ g \circ \left(\bigsqcup_{n< \omega} f_n\right) = \bigsqcup_{n< \omega} (g \circ f_n) \qquad \text{and}\qquad \left(\bigsqcup_{n < \omega} f_n\right) \circ g = \bigsqcup_{n< \omega} (f_n \circ g)\text{.}$$
The composition is called \emph{left strict} if $\bot_{Y,Z} \circ f
= \bot_{X,Z}$ for all arrows $f\colon X \to Y$
.

In our applications, $\cat C$ will mostly be a Kleisli category for a
monad on $\Sets$. Throughout this subsection we assume
that $\cat C$ is a $\cppo$-enriched category.

An endofunctor $H\colon \cat C \to \cat C $ is said to be \emph{locally
  continuous} if for any $\omega$-chain $f_n\colon X \to
Y$, $n < \omega$ in $\cat C(X,Y)$ we have:
\[
H\left(\bigsqcup_{n< \omega} f_n\right)= \bigsqcup_{n< \omega}H(f_n)\text{.}
\]

We are going to make use of the fact that a locally continuous endofunctor
$H$ on $\cat C$ has a \emph{canonical fixpoint}, i.e. whenever its
initial algebra exists it is also its final coalgebra:

\begin{theorem}[\cite{Freyd}]\label{thm:Freyd}
 Let $H \colon \cat{C} \to \cat{C}$ be a locally continuous endofunctor on the \cppo-enriched category $\cat{C}$ whose composition is left-strict.
 If an initial $H$-algebra $\iota \colon HI \congrightarrow I$ exists,
 then $\iota^{-1} \colon I \congrightarrow HI$ is a final $H$-coalgebra.
\end{theorem}

In the sequel, we will be interested in free algebras for a
functor $H$ on $\cat C$ and the free monad $\free{H}$
(cf.~Example~\ref{ex:mnds}.\ref{pt:freemonad}). For this observe that coproducts in $\C$ are always $\cppo$-enriched, i.e.~all copairing maps $[-,-]:\cat C(X,Y) \times \cat C(X', Y) \to \cat C(X+X',Y)$ are continuous; in fact, it is easy to show that this map is continuous in both of its arguments using that composition with the coproduct injections is continuous.

\newcommand{\lemHstar}{  Let $\cat C$ be $\cppo$-enriched with composition left-strict.
  Furthermore, let $H: \cat C \to \cat C$ be locally continuous and assume that all
  free $H$-algebras exist. Then the free monad $\free{H}$ is locally continuous.
}

\begin{proposition}
  \label{lem:Hstar}
\lemHstar
  \end{proposition}

\subsection{Final Coalgebras in Kleisli categories}\label{SSec:Coalg}

In our applications the $\cppo$-enriched
category will be the Kleisli category $\cat C = \Kl(M)$ of a
monad on $\Sets$ and the endofunctors of interest are liftings
$\lift{H}$ of endofunctors $H$ on $\Sets$. It is known that in this
setting a final coalgebra  for the lifting $\lift H$ can be obtained
as a lifting of an initial $H$-algebra (see Hasuo et
al.~\cite{HasuoJS:07}).
 The following result is a variation of Theorem~3.3 in~\cite{HasuoJS:07}:

\begin{theorem}\label{thm:HJS}
Let $M\colon \Sets \to \Sets$ be a monad and $H\colon \Sets \to \Sets$ be a functor such that
\begin{enumerate}[(a)]
 \item $\Kl(M)$ is \cppo-enriched with composition left strict;
 \item $H$ is accessible (i.e., $H$ preserves $\lambda$-filtered colimits for some cardinal $\lambda$) and has a lifting $\lift{H} \colon \Kl(M) \to \Kl(M)$ which is
   locally continuous.
\end{enumerate}
If $\iota \colon H\carrier \congrightarrow \carrier$ is the initial algebra for the functor $H$,
then
\begin{enumerate}
 \item $\J\iota\colon \lift{H}\carrier \to \carrier$ is the initial algebra for the functor $\lift{H}$;
 \item $\J\iota^{-1}\colon \carrier \to \lift{H}\carrier$ is the final coalgebra for the functor $\lift{H}$.
\end{enumerate}
\end{theorem}
The first item follows from Proposition~\ref{prop:liftinginitialalgebra} and the second one follows
from Theorem~\ref{thm:Freyd}. There are two differences with Theorem~3.3 in~\cite{HasuoJS:07}:
\begin{enumerate}[(1)]
\item The functor $H\colon \Sets \to \Sets$ is supposed to preserve $\omega$-colimits rather that being accessible. We use the assumption of accessibility because it guarantees the existence of all free algebras for $H$ and for $\lift{H}$, which implies also that for all $Y \in \Kl(M)$ an initial $\free{\lift{H}}(\Id+Y)$-algebra exists. This property of $\free{\lift{H}}$ will be needed for applying our framework of Section \ref{sec:Theory}.

\item We assume that the lifting $\lift{H} \colon \Kl(M) \to \Kl(M)$ is locally continuous rather than locally monotone. We will need continuity to ensure the double dagger law in Remark~\ref{rmk:doubledagger}. This assumption is not really restrictive since, as explained in Section 3.3.1 of~\cite{HasuoJS:07}, in all the meaningful examples where $\lift{H}$ is locally monotone, it is also locally continuous.
\end{enumerate}

\begin{example}[NDA]\label{ex:NDA} Consider the powerset monad $\Pow$ (Example \ref{ex:mnds}.\ref{pt:powersetmonad}) and the functor $HX = A \times X + 1$ on $\Sets$ (with $1=\{\checkmark\}$). The functor $H$ lifts to $\lift{H}$ on $\Kl(\Pow)$ as follows: for any $f \colon X \to Y$ in $\Kl(\Pow)$ (that is $f \colon X \to \Pow(Y)$ in $\Sets$), $\lift{H}f \colon A\times X + 1 \to A \times Y +1 $ is given by $\lift{H}f( \checkmark ) = \{\checkmark\}$ and $\lift{H}f(\<a,x\>) = \{\<a,y\> \mid y \in f( x )\}$.

  Non-deterministic automata (NDA) over the input alphabet $A$ can be
regarded as coalgebras for the functor
$\lift{H}\colon \Kl(\Pow) \to \Kl(\Pow)$. Consider, on the left, a 3-state NDA, where the only final state is marked by a double circle.
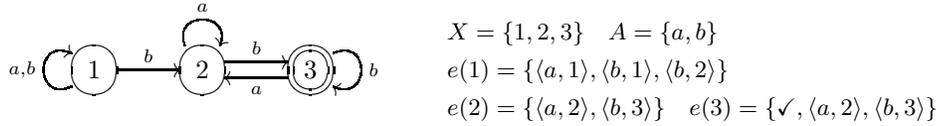
\[
\begin{array}{@{}l@{\qquad}l}
  \xymatrix{
    \state{1} \ar@(ld,ul)[]^{a,b} \ar[r]^-b & \state{2} \ar@<3pt>[r]^{b}
    \ar@(lu,ru)[]^a &
    \fstate{3} \ar@<3pt>[l]^{a} \ar@(ur,dr)[]^b
  } & {\small
  \begin{array}{l}
 X=\{1,2,3\} \quad A = \{a,b\}\\
e(1)=\{ \<a,1\>,
\<b,1\>, \<b,2\>\} \\
e(2) = \{\<a,2\>, \<b,3\> \}\quad
e(3)=\{\checkmark, \<a,2\>,\<b,3\>\}
  \end{array}}
  \end{array}
  \]
  It can be represented as a coalgebra $e \colon X \to
\lift{H}X$, that is a function $e \colon X \to \Pow ( A
\times X +1 )$, given above on the right, which assigns to each state $x\in X$ a set which: contains
$\checkmark$ if $x$ is final; and $\<a,y\>$ for
all transitions $x\xrightarrow{a} y$.
\end{example}

It is easy to see that $M=\Pow$ and $H$ above satisfy the conditions of Theorem \ref{thm:HJS} and therefore both the final $\lift{H}$-coalgebra and the initial $\lift{H}$-algebra are the lifting of the initial algebra for the functor $H X=A \times X +1$, given by $A^*$ with structure $\iota \colon A \times A^{*} +1 \to A^{*}$ which maps $\<a,w\>$ to $aw$ and $\checkmark$ to $\epsilon$.

For an NDA $(X,e)$, the final coalgebra homomorphism $\toFinal{e} \colon X \to A^*$ is the function $X \to \Pow A^*$ that maps every state in $X$ to the language that it accepts. In $\Kl(\Pow)$:
$$
\xymatrix@C=4cm@R=0.5cm{X \ar[dd]_e \ar@{}[ddrr]|{\small
\begin{array}{rcl}
\epsilon \in \toFinal{e}(x) & \Leftrightarrow & \checkmark \in e(x)\\
aw \in \toFinal{e}(x) & \Leftrightarrow & \text{ for some } y\in X, \, (a,y)\in e(x) \text{ and } w\in \toFinal{e}(y) \\
\end{array}} \ar@{-->}[rr]^{\toFinal{e}} & & A^{*} \ar[dd]^{\J\iota^{-1}}\\
\\
A\times X +1 \ar@{-->}[rr]_{A \times \toFinal{e}+1} & & A \times A^* +1
}
$$

\subsection{Monads with Fixpoint Operators}\label{Sec:Elgot}

In order to develop our theory of systems with internal behavior, we will adopt an equational perspective on coalgebras. In the sequel we recall some preliminaries on this viewpoint.

Let $T: \C \to \C$ be a monad on any category $\C$. Any morphism $e: X \to T(X+Y)$ (i.e.~a coalgebra for the functor
$T(\Id+Y)$) may be understood as a system of mutually recursive equations. In our
applications we are interested in the case where $\C = \Kl(M)$ and
 $T = \lift{\free{H}}$ is a (lifted) free monad. As in the example of NDA (Example \ref{ex:NDA}) take $M =
\Pow$ and $HX = 1 + A \times X$. Now, set $TX = A^{*}
+ A^{*} \times X$ and consider the following system of mutually
recursive equations
\[
  x_0  \approx  \{ c, (ab, x_1) \}, \qquad \quad x_1  \approx  \{ d, (a, x_0), (\eps, y) \},
\]
where $x_0, x_1 \in X$ are \emph{recursion variables}, $y \in Y$ is a \emph{parameter} and $a, b, c, d \in A$. A \emph{solution} assigns to each
of the two variables $x_0, x_1$ an element of $\Pow(TY)$ such that
the formal equations $\approx$ become actual identities in $\Kl(\Pow)$:
\[
x_0 \mapsto \{ (aba)^*c, (aba)^*abd, ((aba)^*ab, y)\},
\quad
x_1 \mapsto \{ (aab)^*d, (aab)^*ac, ((aab)^*, y)\}.
\]
Observe that the above system of equations corresponds to an \emph{equation morphism} $e \colon X \to T(X+Y)$ and the solution to a morphism $\gensol e \colon X \to TY$, both in $\Kl(M)$. The property that $\gensol e$ is a solution for $e$ is expressed by the following equation in $\Kl(M)$:
\begin{equation}
\label{diag:fixp}
\gensol e = (\xymatrix@1{
  X \ar[r]^-e
  &
  T(X+Y) \ar[rr]^-{T[\gensol e, \eta^T_Y]}
  &&
  TTY
  \ar[r]^-{\mu^T_Y}
  &
  TY
  }).
\end{equation}
So $e \mapsto \gensol e$ is a \emph{parametrized fixpoint operator}, i.e.~a family of fixpoint operators indexed by parameter sets $Y$.
\begin{remark} \label{rmk:doubledagger} In our
applications we shall need a certain equational property of the operator $e \mapsto \gensol e$: for all $Y \in \C$ and equation morphism $e: X \to T(X+X+Y)$,
the following equation, called \emph{double dagger law}, holds:
\[
e^{\dagger\dagger} = (\xymatrix@1{
  X \ar[r]^-e & T(X+X+Y) \ar[rr]^-{T(\nabla_X + Y)} && T(X+Y)
})^\dagger.
\]
This and other laws of parametrized fixpoint operators have been
studied by Bloom and \'Esik in the context of \emph{iteration
  theories}~\cite{be93}. 
A closely related notion is that of \emph{Elgot monads}~\cite{amv10,amv11}.
%
%
\end{remark}
\begin{example}[Least fixpoint solutions]\label{ex:LfpSolCPO}
  Let $T: \C \to \C$ be a locally continuous monad on the \cppo-enriched category $\C$. Then $T$ is equipped wi th a parametrized fixpoint operator obtained by taking least fixpoints: given a morphism $e: X \to
  T(X+Y)$ consider the function $\Phi_e$ on $\C(X,TY)$ given by $\Phi_e(s) = \mu^T_Y \circ T[s,\eta^T_Y] \circ e$. Then $\Phi_e$ is continuous and we take $\lsol e$ to be the least fixpoint of $\Phi_e$. Since $\lsol e= \Phi_e(\lsol e)$, equation~\refeq{diag:fixp} holds, and it follows from the argument in Theorem~8.2.15 and Exercise~8.2.17 in~\cite{be93} that the operator $e \mapsto \lsol e$ satisfies the axioms of iteration theories (or Elgot monads, respectively). In particular the double dagger law holds for the least fixpoint operator $e \mapsto \lsol e$.
\end{example}

\section{Motivating examples}\label{SSec:Mot}

The work of~\cite{HasuoJS:07} bridged a gap in the theory of coalgebras: for certain functors, taking the final coalgebra directly in $\Sets$ does not give the right notion of equivalence. For instance, for NDA, one would obtain bisimilarity instead of language equivalence. The change to Kleisli categories allowed the recovery of the usual language semantics for NDA and, more generally, led to the development of \emph{coalgebraic trace semantics}.

In the Introduction we argued that there are relevant examples for which this approach still yields the unwanted notion of equivalence, the problem being that it does not consider the extra algebraic structure on the label set. In the sequel, we motivate the reader for the generic theory we will develop by detailing two case studies in which this phenomenon can be observed: NDA with $\epsilon$-transitions and NDA with word transitions. Later on, in Example \ref{Sec:MazurTraces}, we will also consider Mazurkiewicz traces~\cite{Mazurkiewicz77}.

\paragraph{NDA with $\epsilon$-transitions.} In the world of automata, $\epsilon$-transitions are considered in order to enable easy composition of automata and compact representations of languages. These transitions are to be interpreted as the empty word when computing the language accepted by a state. Consider, on the left, the following simple example of an NDA with $\epsilon$-transitions, where states $x$ and $y$ just make $\epsilon$ transitions. The intended semantics in this example is that  all states accept words in $a^*$.

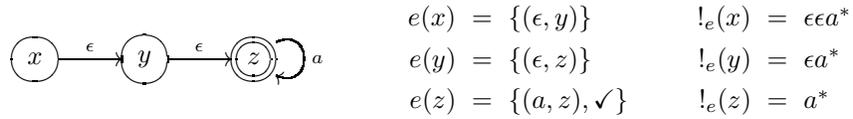
$$
\begin{array}{cll}
  \xymatrix{
    \state{x} \ar[r]^{\epsilon}  & \state{y} \ar[r]^{\epsilon} & \fstate{z} \ar@(ur,dr)[]^{a} }
 &\qquad
\begin{array}{rcl}
 e(x) &=& \{(\epsilon, y)\} \\
 e(y) &=& \{(\epsilon, z)\} \\
 e(z) &=& \{(a,z), \checkmark \}
\end{array}\qquad
\begin{array}{rcl}
 \toFinal{e}(x) &=& \epsilon \epsilon a^{*} \\
 \toFinal{e}(y) &=& \epsilon a^{*} \\
 \toFinal{e}(z) &=&  a^{*}
\end{array}

\end{array}
$$
Note that, more explicitly, these are just NDA where the alphabet has a distinguished symbol $\epsilon$. So, they are coalgebras for the functor $\lift{H + \Id}\colon \Kl(\Pow) \to \Kl(\Pow)$ (where $H$ is the functor of Example~\ref{ex:NDA}), i.e.~functions $e \colon X \to \Pow( (A \times X + 1) +X) \cong \Pow( (A+1) \times X + 1)$, as made explicit for the above automaton in the middle.

The final coalgebra for $\lift{H+\Id}$ is simply $(A+1)^{*}$ and the final map $\toFinal{e} \colon X \to (A+1)^{*}$ assigns to each state the language in $(A+1)^{*}$ that it accepts. However, the equivalence induced by $\toFinal{e}$ is too fine grained: for the automata above, $\toFinal{e}$ maps $x$, $y$ and $z$ to three different languages (on the right), where the number of $\epsilon$ plays an explicit role, but the intended semantics should disregard $\epsilon$'s.


\paragraph{NDA with word transitions.} This is a variation on the motivating example of the introduction: instead of languages, transitions are labeled by words\footnote{More generally, one could consider labels from an arbitrary monoid.}. Formally, consider again the functor $H$ from Example~\ref{ex:NDA}. Then NDA with word transitions are coalgebras for the functor $\lift{\free{H}} \colon \Kl(\Pow) \to \Kl(\Pow)$, that is, functions $e \colon X \to \Pow(A^{*} \times X + A^{*})\cong \Pow(A^{*}\times (X+1))$. We observe that they are like NDA  but (1) transitions are labeled by words in $A^{*}$, rather than just symbols of the alphabet $A$, and (2) states have associated output languages, rather than just $\checkmark$. We will draw them as ordinary automata plus an arrow $\stackrel{L}{\To}$ to denote the output language of a state (no $\To$ stands for the empty language). For an example, consider the following word automaton and associated transition function $e$.
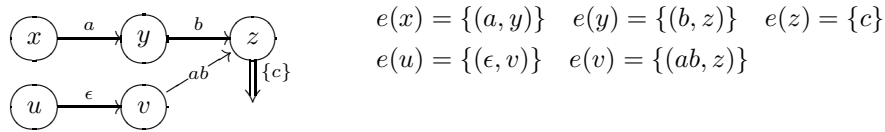
$$\begin{array}{cll}
\xymatrix@R=.3cm{
    \state{x} \ar[r]^{a}  & \state{y} \ar[r]^{b} & \state{z} \ar@{=>}[d]^{\{c\}} \\
    \state{u} \ar[r]^{\epsilon}  & \state{v} \ar[ru]|{ab} &  \\}
  &\qquad
\begin{array}{l}
 e(x) = \{(a,y)\} \quad
 e(y) = \{(b,z)\} \quad
 e(z) = \{c\}\\
 e(u) = \{(\epsilon,v)\}\quad
  e(v) = \{(ab,z)\}
 \end{array}
\end{array}
$$
The semantics of NDA with word transitions is given by languages over $A$, obtained by concatenating the words in the transitions and ending with a word from the output language. For instance, $x$ above accepts word $abc$ but not $ab$.

However, if we consider the final coalgebra semantics we again have a mismatch.
The initial $\free{H}$-algebra has carrier  $(A^{*})^{*} \times A^{*}$ that can be represented as the set of non-empty lists of words over $A^*$, where $(A^{*})^{*}$ indicates possibly empty lists of words. Its structure $\iota \colon A^{*} \times ((A^{*})^{*} \times A^{*}) + A^{*} \to (A^{*})^{*} \times A^{*}$ maps $w$ into $(\tnil,w)$ and $(w',(l,w))$ into $(w'::l,w)$. Here, we use $\tnil$ to denote the empty list and $::$ is the append operation. By Theorem \ref{thm:HJS},
the final $\lift{\free{H}}$-coalgebra has the same carrier and structure $\J\iota^{-1}$. The final map, as a function $\toFinal{e} \colon X \to \Pow((A^{*})^{*} \times A^{*})$, is then defined by commutativity of the following square (in $\Kl(\Pow)$):
\begin{equation}
\label{eq:finalmap_wordaut}
\vcenter{
\small
\xymatrix@R=17pt@C=2.7cm{
X \ar@{-->}[rr]^{\toFinal{e}} \ar[dd]_{e} \ar@{}[rrdd]|{\small \begin{array}{rcl}
(\tnil, w) \in \toFinal{e}(x) & \Leftrightarrow & w \in e(x)  \\
(w::l,w') \in \toFinal{e}(x) & \Leftrightarrow & \exists_y\ (w,y) \in e(x) \text{ and } (l,w') \in \toFinal{e}(y).
\end{array}
}
 & &
(A^{*})^{*} \times A^{*} \ar[dd]^{\J\iota^{-1}} \\
 &  & \\
A^{*} \times X+A^{*} \ar@{-->}[rr]_{\id_{A^{*}} \times \toFinal{e} + \id_{A^{*}}} & & A^{*} \times ((A^{*})^{*} \times A^{*}) + A^{*}
}
}
\end{equation}

Once more, the semantics given by $\toFinal{e}$ is too fine grained: in the above example, $\toFinal{e}(x)=\{([a,b],c)\}$ and $\toFinal{e}(u)=\{([\epsilon,ab],c)\}$ whereas the intended semantics would equate both $x$ and $u$, since they both accept the language $\{abc\}$.

\medskip

Note that any NDA can be regarded as word automaton. Recall the natural transformation $\kappa \colon \lift{H} \To \lift{\free{H}}$ defined in Example \ref{ex:mnds}.\ref{pt:freemonad}: for the functor $\lift{H}$ of NDA, $$\kappa_X \colon A \times X +1 \to A^{*}\times X + A^{*}$$ maps any pair $(a,x)\in A\times X$ into $\{(a,x)\}\in \Pow(A^{*}\times X + A^{*})$ and $\checkmark \in 1$ into $\{ \epsilon \} \in \Pow(A^{*}\times X + A^{*})$. Composing an NDA $e \colon X \to \lift{H} X$ with $\kappa_X\colon \lift{H} X \to \lift{\free{H}} X$, one obtains the word automaton $\kappa_X \circ e$.

In the same way, every NDA with $\epsilon$-transitions can also be seen as a word automaton by postcomposing with the natural transformation $[\kappa, \eta] \colon \lift{H+\Id} \To \lift{\free{H}}$. Here, $\eta\colon \Id \To \lift{\free{H}}$ is the unit of the free monad
$\lift{\free{H}}$ defined on a given set $X$ below (the multiplication $\mu\colon \lift{\free{H}}\lift{\free{H}} \To \lift{\free{H}}$ is shown on the right).
$$
\begin{array}{l@{\qquad}l}
\eta_X \colon   X  \to      A^{*} \times X+A^{*}& \mu_X \colon   A^{*} \times (( A^{*} \times X  + A^{*} ) +  A^*  \to A^{*} \times X +  A^{*}\\
                x  \mapsto  \{(\epsilon, x)\}&               (w, (w',x))  \mapsto  \{(w \cdot w', x) \}\quad
               (w,  w')     \mapsto   \{ w\cdot w' \}\\
&  w \mapsto \{w\}
\end{array}
$$
In the next section, we propose to define the semantics of $\lift{\free{H}}$-coalgebras via a canonical fixpoint operator rather than with the final map which as we saw above might yield unwanted semantics. Then, using the observation above, the semantics of $\lift{H}$-coalgebras and $\lift{H+\Id}$-coalgebras will be defined by embedding them into $\lift{\free{H}}$-coalgebras via the natural transformations $\kappa$ and $[\kappa,\eta]$ described above.
\section{Canonical Fixpoint Solutions}\label{sec:Theory}
In this section we lay the foundations of our approach. A construction is introduced assigning canonical solutions to coalgebras seen as equation morphisms (\emph{cf.} Section \ref{Sec:Elgot}) in a \cppo-enriched setting. We will be working under the following assumptions.
\begin{assumption}\label{ass:CPPOenrichTcont}
  Let $\C$ be a \cppo-enriched category with coproducts and  composition left-strict.
  Let $T$ be a locally continuous monad on $\C$ such
  that, for all object $Y$, an initial algebra for $T(\Id +Y)$ exists.
\end{assumption}

As seen in Example \ref{ex:LfpSolCPO}, in this setting an equation morphism $e: X \to T(X+Y)$ may be given the least solution. Here, we take  a different approach, exploiting the initial algebra-final coalgebra coincidence of Theorem~\ref{thm:Freyd}.


For every parameter object $Y \in \C$, the endofunctor $T(\Id +Y)$ is
a locally continuous monad because it is the composition of $T$ with the (locally continuous) exception monad $\Id +Y$. Thus, by Theorem \ref{thm:Freyd} applied to $T(\Id +Y)$, the initial $T(\Id +Y)$-algebra $\iota_Y : T(\carrier_Y + Y)\xrightarrow{\cong} \carrier_Y$ yields a final $T(\Id +Y)$-coalgebra $\iota_Y^{-1} : \carrier_Y \xrightarrow{\cong} T(\carrier_Y + Y)$. This allows us to associate with any equation morphism $e: X \to T(X+Y)$ a canonical morphism of type $X \to TY$ as in the following diagram.
\begin{equation}\label{diag:canonicalSol}
\vcenter{
    \xymatrix@R=10pt{
    X \ar@{-->}[rr]^{\toFinal{e}} \ar[dd]_{e} & & \carrier_Y \ar@/_/[dd]_{\iota_Y^{-1}} \ar@{-->}[rr]^{\fromInit} & & TY\\
     & & & & TTY \ar[u]_{\mu^T_Y} \\
    T(X+Y)  \ar@{-->}[rr]_{T(\toFinal{e} + \id_Y)} & & T(\carrier_Y +Y)
    \ar@/_/[uu]_{\iota_Y}
    \ar@{-->}[rr]_{T( \fromInit + \id_Y)} & & T(TY+Y) \ar[u]_{T[\id_{T Y},\eta^T_Y]}
    }
}
\end{equation}
In \eqref{diag:canonicalSol}, the map $\toFinal{e}\: X \to \carrier_Y$ is the unique morphism of  $T(\Id +Y)$-coalgebras given by finality of  $\iota_Y^{-1} \: \carrier_Y \to T(\carrier_Y + Y)$, whereas $\fromInit: \carrier_Y \to TY$ is the unique morphism of $T(\Id +Y)$-algebras given by initiality of $\iota_Y \: T(\carrier_Y + Y)\to \carrier_Y$.

We call the composite $\fromInit \circ \toFinal{e}\: X \to TY$ the \emph{canonical fixpoint solution} of $e$. In the following we check that the canonical fixpoint solution is indeed a solution of $e$, in fact, it coincides with the least solution.

\newcommand{\proplfpcfpsolution}{Given a morphism $e\: X \to T(X+Y)$, then the least solution of $e$ as in Example \ref{ex:LfpSolCPO} is the canonical fixpoint solution: $\cansol e =\fromInit \circ \toFinal{e}\: X \to TY$ as in \eqref{diag:canonicalSol}.}
\begin{proposition} \label{prop:lfp=cfpsolution}
\proplfpcfpsolution
\end{proposition}

As recalled in Example \ref{ex:LfpSolCPO}, the least fixpoint operator $e \mapsto \lsol e$ satisfies the double dagger law. Thus Proposition \ref{prop:lfp=cfpsolution} yields the following result\footnote{The equality of least and canonical fixpoint solutions can be used to state a stronger result, namely that canonical fixpoint solutions satisfy the axioms of iteration theories (\emph{cf.} Example \ref{ex:LfpSolCPO}). However, the double dagger law is the only property that we need here, explaining the statement of Corollary \ref{cor:cansolElgot}.}

\begin{corollary}\label{cor:cansolElgot} Let $\C$ and $T \: \C \to \C$ be as in Assumption \ref{ass:CPPOenrichTcont}. Then the canonical fixpoint operator $e \mapsto \cansol e$ associated with $T$ satisfies the double dagger law.
\end{corollary}

We now introduce a factorisation result on the operator $e \mapsto \cansol e$, which is useful for comparing solutions provided by different monads connected via a monad morphism.
\newcommand{\propfact}{Suppose that $T$ and $T'$ are monads on $\C$ satisfying Assumption \ref{ass:CPPOenrichTcont} and $\gamma \: T \Rightarrow T'$ is a monad morphism. For any morphism $e : X \to T(X+Y)$:
\[ \gamma_Y \circ \cansol e = \cansol {(\gamma_{X+Y} \circ e)} : X \to T' Y,\]
where $\cansol e$ is provided by the canonical fixpoint solution for $T$ and $\cansol {(\gamma_{X+Y} \circ e)}$ by the one for $T'$.}
\begin{proposition}[Factorisation Lemma] \label{prop:factorizationLemma}
\propfact
\end{proposition}
\subsection{A Theory of Systems with Internal Behavior}

We now use canonical fixpoint solutions to provide an abstract theory of systems with internal behavior, that we will later instantiate to the motivating examples of Section \ref{SSec:Mot}. Throughout this section, we will develop our framework for the following ingredients.

\begin{assumption} \label{ass:TheoryHSystems} Let $\C$ be a \cppo-enriched category with coproducts and composition left-strict and let $\GF \colon \C \to \C$ be a locally continuous functor for which all free $\GF$-algebras exist. Consider the following two monads derived from $\GF$:
\begin{itemize}
  \item the free monad $\free{\GF} \colon \C \to \C$ (\emph{cf.} Example \ref{ex:mnds}.\ref{pt:freemonad}), for which we suppose that an initial $\free{\GF}(\Id +Y)$-algebra exists for all $Y\in\C$;
  \item for a fixed $X \in \C$, the exception monad $\GF X + \Id \colon \C \to \C$ (\emph{cf.} Example \ref{ex:mnds}.\ref{pt:exceptionmonad}), for which we suppose that an initial $\GF  X+\Id +Y$-algebra exists for all $Y\in\C$.
\end{itemize}
\end{assumption}
In the next proposition we verify that the construction introduced in the previous section applies to the two monads of Assumption \ref{ass:TheoryHSystems}.
 \newcommand{\propass}{Let $\C$, $\GF $,  $\free{\GF }$ and  $\GF X + \Id$ be as in Assumption \ref{ass:TheoryHSystems}. Then $\C$ and the monads $\free{\GF }\colon \C \to \C$ and $\GF X + \Id \colon \C \to \C$ satisfy Assumption \ref{ass:CPPOenrichTcont}. Thus both $\free{\GF }$ and  $\GF X + \Id$ are monads with canonical fixpoint solution (which satisfy the double dagger law by Corollary \ref{cor:cansolElgot}).}
 \begin{proposition}\label{prop:ass2->ass1}
 \propass
 \end{proposition}
   To avoid ambiguity, we denote with $e \mapsto \cansol e$ the canonical fixpoint operator associated with $\free{\GF }$ and with $e \mapsto \altsol e$ the one associated with $\GF X + \Id$.

We will employ the additional structure of those two monads for the analysis of \emph{$\GF $-systems with internal transitions}. An $\GF $-system is simply an $\GF $-coalgebra $e\: X \to \GF X$, where we take the operational point of view of seeing $X$ as a space of states and $\GF $ as the transition type of $e$. An $\GF $-system with internal transitions is an $(\GF  + \Id)$-coalgebra $e \: X \to \GF X +X$, where the component $X$ of the codomain is targeted by those transitions representing the internal (non-interacting) behavior of system $e$.

A key observation for our analysis is that $\GF $-systems---with or without internal transitions---enjoy a standard representation as $\free{\GF }$-systems, that is, coalgebras of the form $e\: X \to \free{\GF }X$.
\begin{definition}[$\GF $-systems as $\free{\GF }$-systems] \label{def:reprHtofreeH} Let $\kappa : \GF  \to \free{\GF }$ be as in Example~\ref{ex:mnds}.\ref{pt:freemonad}. We introduce the following encoding $e \mapsto \bar{e}$ of $\GF $-systems and $\GF $-systems with internal transitions as $\free{\GF }$-systems.
\begin{itemize}
  \item Given an $\GF $-system $e \colon X \to \GF X$, define $\bar{e} \colon X \to \free{\GF }X$ as
\[
\bar{e} : \xymatrix@1{X \ar[r]^-{e} & \GF X \ar[r]^{ \kappa_X} & \free{\GF }X.}
\]
\item Given an $\GF $-system with internal transitions $e : X \to \GF X + X$, define $\bar{e} : X \to \free{\GF }X$ as $\bar{e} : \xymatrix@1{X \ar[r]|-{\ e\ } & \GF X + X \ar[rr]|-{\ [ \kappa_X, \eta^{\free{\GF }}_X]\ } & & \free{\GF }X}$.
  \end{itemize}
\end{definition}
Thus $\GF $-systems (with or without internal transitions) may be seen as equation morphisms $X \to \free{\GF}(X+0)$ for the monad $\free{\GF }$ (with the initial object $Y = 0$ as parameter), with solutions by canonical fixpoint ({\em cf.}~Section~\ref{Sec:Elgot}). This will allow us to achieve the following.
\begin{enumerate}[label={(\arabic*)}]
  \item[\bf \S 1] We supply a uniform trace semantics for $\GF $-systems, possibly with internal transitions, and $\free{\GF }$-systems, based on the canonical fixpoint solution operator of $\free{\GF }$. 
  \item[\bf \S 2] We use the canonical fixpoint operator of $\GF X + \Id$ to transform any $\GF $-system $e: X \to \GF X + X$ with internal transitions into an $\GF $-system $\epselim{e} : X \to \GF X$ without internal transitions. 
  \item[\bf \S 3] We prove that the transformation of {\bf \S 2} is sound with respect to the semantics of {\bf \S 1}.
\end{enumerate}
\paragraph{\bf \S 1: Uniform trace semantics.}  The canonical fixpoint semantics of $\GF $-systems, with or without internal transitions, and $\free{\GF }$-systems is defined as follows.
 \begin{definition}[Canonical Fixpoint Semantics] \label{def:canfixSem}
\begin{itemize}
  \item For an $\free{\GF }$-system $e\colon X \to\free{\GF }X$, its semantics $\bb{e} \colon X \to \free{\GF }0$ is defined as $\cansol {e}$
(note that $e$ can be seen as  an equation morphism for $\free{\GF }$ on parameter $Y = 0$).
\item For an $\GF $-system $e:X \to\GF X$, its semantics $\bb{e} \colon X \to \GF 0$ is defined as $\cansol {\bar{e}}= \cansol {(\kappa_X \circ e)}$.
\item For an $\GF $-system with internal transitions $e:X \to \GF X+X$, its semantics $\bb{e} \colon X \to \GF 0$ is defined as $\cansol {\bar{e}}= \cansol {([\kappa_X,\eta_X^{\free{\GF }}] \circ e)}$.
\end{itemize}
\end{definition}
 The underlying intuition of Definition \ref{def:canfixSem} is that canonical fixpoint solutions may be given an operational understanding. Given an $\free{\GF }$-system $e\: X \to \free{\GF }X$, its solution $\cansol{e}\: X \to \free{\GF }0$ is formally defined as the composite $\fromInit \circ \toFinal{e}$ (\emph{cf.} \eqref{diag:canonicalSol}): we can see the coalgebra morphism $\toFinal{e}$ as a map that gives the \emph{behavior} of system $e$ without taking into account the structure of labels and the algebra morphism $\fromInit$
as evaluating this structure, e.g. flattening words of words, using the initial algebra $\mu_0 \: \free{\GF }\free{\GF }0 \to \free{\GF }0$ for the monad $\free{\GF}$. In particular, the action of $\fromInit$ is what makes our semantics suitable for modeling ``algebraic'' operations on internal transitions such as $\epsilon$-elimination, as we will see in concrete instances of our framework.

\begin{remark} The canonical fixpoint semantics of Definition \ref{def:canfixSem} encompasses the framework for traces in \cite{HasuoJS:07}, where the semantics of an $\GF $-system $e : X \to \GF X$---without internal transitions---is defined as the unique morphism $\toFinal{e}$ from $X$ into the final $\GF $-coalgebra $\free{\GF }0$. Indeed, using finality of $\free{\GF }0$, it can be shown that $\toFinal{e} = \bb{e}$. Theorem~\ref{th:comparewithHJS} below guarantees compatibility with Assumption~\ref{ass:TheoryHSystems}.
\end{remark}
The following result is instrumental in our examples and in comparing our theory with the one developed in \cite{HasuoJS:07} for trace semantics in Kleisli categories.
\newcommand{\thcomparewithHJS}{Let $M \: \Sets \to \Sets$ be a monad and $H  \: \Sets \to \Sets$ be a functor satisfying the assumptions of Theorem \ref{thm:HJS}, that is:
\begin{enumerate}[(a)]
 \item $\Kl(M)$ is \cppo-enriched and composition is left strict;
 \item $H$ is accessible and has a locally continuous lifting $\lift{H} \colon \Kl(M) \to \Kl(M)$.
\end{enumerate}
Then $\Kl(M)$, $\lift{H}$, $\free{\lift{H}}$ and $\lift{H}\J  X + \Id$ (for a given set $X$) satisfy Assumption \ref{ass:TheoryHSystems}.}
\begin{theorem}\label{th:comparewithHJS}
\thcomparewithHJS
\end{theorem}
\begin{example}[Semantics of NDA with word transitions]
In Section \ref{SSec:Mot}, we have modeled NDA with word transitions as $\lift{\free{H }}$-coalgebras on $\Kl(M)$, where $H$ and $M$ are defined as for NDA (see Example \ref{ex:NDA}). By Proposition~\ref{prop:liftingfreemonad}, $\lift{\free{H}} = \free{\lift{H}}$ and thus, by virtue of Theorem \ref{th:comparewithHJS}, $\lift{\free{H}}$ satisfies Assumption \ref{ass:TheoryHSystems}. Therefore we can define the semantics of NDA with word transitions $e \colon X \to \Pow (A^*\times X + A^*)$ via canonical fixpoint solutions as $\bb{e} = \cansol{e} = \fromInit\after \toFinal{e}$:
\begin{eqnarray}
\vcenter{\small
    \xymatrix@R=14pt{
    X \ar@{-->}[rr]^{\toFinal{e}} \ar[dd]_{e} & & (A^{*})^{*} \times A^{*} \ar@/_/[dd]^{\cong} \ar@{-->}[rrr]^{\fromInit}
    \ar@{}[rrrdd]|{\small \begin{array}{rcl}
\fromInit (\tnil, w) & = & \{w\} \\
\fromInit (w::l,w')  & = & \{wu \mid u\in \fromInit(l,w')\}
\end{array}}
    & & & A^{*}\\
     & & & &  & & & & & & \\
    A^{*} \times X + A^{*}  \ar@{-->}[rr]_-{\id \times \toFinal{e} +\id} & & A^{*} \times ((A^{*})^{*} \times A^{*}) + A^{*}\ar@/_/[uu]
    \ar@{-->}[rrr]_-{{\id \times \fromInit + \id}} & & & A^{*} \times A^{*} + A^{*} \ar[uu]_{\mu_0} & & & &
    }
}
\end{eqnarray}
Observe that the above diagram is just \eqref{diag:canonicalSol} instantiated with $T=\lift{\free{H }}$ and $Y=0$. Moreover, this diagram is in $\Kl(\Pow)$ and hence the explicit definition of $\cansol e$ as a function $X \to \Pow(A^*)$ is given by $\cansol e(x) = \bigcup \Pow(\fromInit)(\toFinal{e} (x))$.

Both $\toFinal{e}$ and $\fromInit$ can be defined uniquely by the commutativity of the above diagram.
We have already defined $\toFinal{e}$ in diagram~\eqref{eq:finalmap_wordaut} and the definition of $\fromInit$ is given in the right-hand square of the above diagram. The isomorphism in the middle and $\mu_0$ were defined in Section~\ref{SSec:Mot}.

Using the above formula $\cansol e(x) = \bigcup \Pow (\fromInit) (\toFinal{e}(x))$ we now have the semantics of $e$:
\begin{equation}\label{eq:language}
w \in \cansol{e}(x) \Leftrightarrow  
\begin{array}[t]{l}
w \in {e}(x)\ \textbf{ or }\\
\exists_{y\in X, w_1,w_2\in A^*}\ (w_1, y) \in {e}(x), w_2 \in \cansol{e}(y) \text{ and } w=w_1w_2.
\end{array}
\end{equation}
This definition is precisely the language semantics: a word $w$ is accepted by a state $x$ if there exists a decomposition $w=w_1\cdots w_n$ such that $\xymatrix@C=.4cm{x \ar[r]^-{w_1} & y_1 \ar[r]^-{w_2}& \cdots \ar[r]^-{w_{n-1}}& y_{n-1} \ar@{=>}[r]^-{w_n}&}$.
Take again the automaton of the motivating example. We can calculate the semantics and observe that we now get exactly what was expected: $ \cansol{e}(u)= \cansol{e}(v)$.
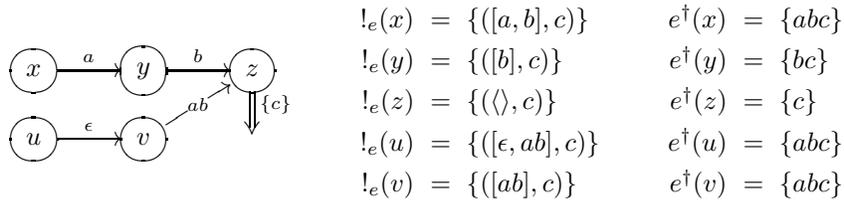
$$\begin{array}{c@{\qquad}l@{\qquad}l}
\vcenter{\xymatrix@R=.3cm{
    \state{x} \ar[r]^{a}  & \state{y} \ar[r]^{b} & \state{z} \ar@{=>}[d]^{\{c\}} \\
    \state{u} \ar[r]^{\epsilon}  & \state{v} \ar[ru]|{ab} &  \\}}
  &
\begin{array}{rcl}
 \toFinal{e}(x) &=& \{([a,b],c)\} \\
 \toFinal{e}(y) &=& \{([b],c)\} \\
 \toFinal{e}(z) &=&  \{(\tnil,c)\} \\
  \toFinal{e}(u) &=& \{([\epsilon,ab],c)\} \\
  \toFinal{e}(v) &=& \{([ab],c)\} \\
\end{array}
&
\begin{array}{rcl}
 \cansol{e}(x) &=& \{abc\} \\
 \cansol{e}(y) &=& \{bc\} \\
 \cansol{e}(z) &=&  \{c\} \\
 \cansol{e}(u) &=& \{abc\} \\
 \cansol{e}(v) &=& \{abc\} \\
\end{array}
\end{array}
$$
The key role played by the monad structure on $A^*$ can be appreciated by comparing the graphs of $\toFinal{e}$ and $\cansol{e} = \fromInit \circ \toFinal{e}$ as in the example above. The algebra morphism $\fromInit \: (A^*)^*\times A^* \to A^*$ maps values from the initial algebra $(A^*)^*\times A^*$ for the \emph{endofunctor} $\lift{\free{H }}$ into the initial algebra $A^*$ for the \emph{monad} $\lift{\free{H }}$: its action is precisely to take into account the additional equations encoded by the algebraic theory of the monad $\lift{\free{H }}$. For instance, we can see the mapping of $ \toFinal{e}(u) = \{([\epsilon,ab],c)\}$ into the word $abc$ as the result of concatenating the words $\epsilon$, $ab$, $c$ and then quotienting out of the equation $\epsilon abc = abc$ in the monoid $A^*$.
\end{example}
\begin{remark}[Multiple Solutions]\label{rm:multiple}
The canonical solution $\cansol e$ is not the unique solution. Indeed, the uniqueness of $\toFinal{e}$ in the left-hand square and of $\fromInit$ in the right-hand square of the diagram above does not imply the uniqueness of $e^\dagger$. To see this, take for instance the automaton
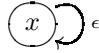
\[
  \xymatrix{
    \state{x}  \ar@(ur,dr)[]^{\epsilon}
  }
  \]
Both $s(x)=\emptyset$ and $s'(x)=A^{*}$ are solutions. The canonical one is the least one, i.e., $e^{\dagger}(x)=s(x)=\emptyset$.
\end{remark}
\begin{example}[{Semantics of NDA with $\epsilon$-transitions}]\label{ex:epselim}
NDA with $\epsilon$-transitions are modeled as $\lift{H + \Id}$-coalgebras on $\Kl(M)$, where $H$ and $M$ are defined as for NDA (see Example~\ref{ex:NDA}). We can define the semantics of NDA with $\epsilon$-transitions via canonical fixpoint solutions as $\bb{e} = \cansol{\bar{e}}$, where $\bar{e}$ is the automaton with word transitions corresponding to $e$ (see Definition~\ref{def:reprHtofreeH}). The first example in Section~\ref{SSec:Mot}
would be represented as follows,
$$
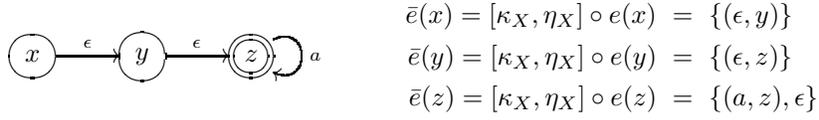
\begin{array}{cll}
  \xymatrix{
    \state{x} \ar[r]^{\epsilon}  & \state{y} \ar[r]^{\epsilon} & \fstate{z} \ar@(ur,dr)[]^{a} }
 &\qquad
\begin{array}{rcl}
\bar e(x) =  [\kappa_X,\eta_X] \circ e(x) &=& \{(\epsilon, y)\} \\
\bar e(y) =  [\kappa_X,\eta_X] \circ e(y) &=& \{(\epsilon, z)\} \\
\bar e(z) =  [\kappa_X,\eta_X] \circ e(z) &=& \{(a,z), \epsilon \}
\end{array}
\end{array}
$$
where $\eta$ and $\kappa$ are defined as at the end of Section~\ref{SSec:Mot}.
By using \eqref{eq:language}, it can be easily checked that the semantics $\bb{e}=\cansol{\bar{e}} \colon X \to \Pow A^{*}$ maps  $x$, $y$ and $z$ into $a^{*}$.
\end{example}
\paragraph{\bf \S 2: Elimination of internal transitions.}
We view an $\GF $-system $e \: X \to \GF X + X$ with internal transitions as an equation morphism for the monad $\GF X + \Id$, with parameter $Y = 0$. Thus we can use the canonical fixpoint solution of $\GF X + \Id$ to obtain an $\GF $-system $\altsol e\: X \to \GF X + 0 = \GF X$, which we denote by $\epselim{e}$. The construction is depicted below.
\begin{eqnarray}\label{diag:cansolHX+Id}
\vcenter{
    \xymatrix@R=10pt@C=20pt{
      X  
      \ar@{-->}[rr]^{\toFinal{e}} \ar[dd]_{e} & & \N \times \GF X \ar@/^/[dd]_{\cong} \ar@{-->}[rr]^{\fromInit} & & \GF X \ar@{<-} `u[l] `[llll]_{\epselim{e}\ \stackrel{\mathrm{def}}=\ \altsol e} [llll] \\
      & & & &  
      \\
      \GF X + X  \ar[rr]_-{\id_{\GF X} +\toFinal{e}} & & \GF X + \N \times \GF X \ar@/^/[uu]
      \ar[rr]_{ \id_{\GF X} + \fromInit} & & \GF X + \GF X\ar[uu]_{\mu_0 = \nabla_{\GF X}}\\
    }
}
\end{eqnarray}
\begin{example}[$\epsilon$-elimination]
Using the automaton of Example~\ref{ex:epselim}, we can perform $\epsilon$-elimination, as defined in~\eqref{diag:cansolHX+Id}, using the canonical solution for the monad $\lift{H}\J X +\Id$:
\begin{eqnarray*}\small
\vcenter{
    \xymatrix@R=10pt@C=5pt{
    X  \ar@{-->}[rr]^{\toFinal{e}} \ar[dd]_{e} & & \N \times (A\times X +1) \ar@/^/[dd]_{\cong} \ar@{-->}[rr]^{\fromInit} & & (A\times X +1)  \\
     & & & &  \\
     (A\times X +1) + X  \ar[rr]_-{\id  + \toFinal{e}} & & (A\times X +1) + \N \times (A\times X +1)\ar@/^/[uu]
    \ar[rr]_-{\id + \fromInit} & & (A\times X +1) + (A\times X +1) \ar[uu]_{\mu_0 = \nabla}
    }
}
\end{eqnarray*}
We obtain the following NDA $\epselim{e} \stackrel{\mathrm{def}}=\fromInit \circ\toFinal{e} \colon X \to A \times X +1$.
$$\small
\begin{array}{ccc}
\begin{array}{lcl}
 \toFinal{e}(x) &=& \{(2,a,z), (2,\checkmark)\} \\
 \toFinal{e}(y) &=& \{(1,a,z), (1,\checkmark)\} \\
 \toFinal{e}(z) &=& \{(0,a,z), (0,\checkmark) \}
\end{array}\quad
&
\begin{array}{lcl}
 \epselim{e}(x) &=& \{(a,z), \checkmark\} \\
  \epselim{e}(y) &=& \{(a,z), \checkmark\} \\
  \epselim{e} (x) &=& \{(a,z), \checkmark \}
\end{array}
\quad
&
  \xymatrix@C=.5cm{
    \fstate{x}  & \fstate{y} \ar[r]^{a} & \fstate{z} \ar@(ur,dr)[]^{a} \ar@{<-} `u[l] `[ll]_a [ll] }
    \end{array}
$$
The semantics $\bb{\epselim e}$ is defined as $\cansol{\overline{\epselim e}}$, where $\overline{\epselim e}= \kappa_X \after \epselim e$ is the representation of the NDA $\epselim e$ as an automaton with word transitions (Definition \ref{def:reprHtofreeH}). It is immediate to see, in this case, that $\bb{\epselim e}=\bb{e}$. This fact is an instance of Theorem \ref{th:epselimsound} below.
\end{example}
\begin{remark} Note that $\epsilon$-elimination was recently defined using a trace operator on a Kleisli category~\cite{Hasuo06,SW13,Asada}. These works are based on the trace semantics of Hasuo et al.~\cite{HasuoJS:07} and tailored for $\epsilon$-elimination. They do not take into account any algebraic structure of the labels and are hence not applicable to the other examples we consider in this paper. \end{remark}
\paragraph{\bf \S 3: Soundness of $\epsilon$-elimination.}
We now formally prove that the canonical fixpoint semantics of $e$ and $\epselim{e}$ coincide. To this end, first we show how the construction $e \mapsto \epselim{e}$ can be expressed in terms of the canonical fixpoint solution of $\free{\GF }$. This turns out to be an application of the factorisation lemma (Proposition \ref{prop:factorizationLemma}), for which we introduce the natural transformation $\pi \: \GF X + \Id \To \free{\GF }(X+\Id)$ defined at $Y \in \C$ by
\[
\pi_Y \colon \xymatrix{\GF X + Y \ar[rr]^-{[\kappa_X,\ \eta^{\free{\GF }}_Y]} && \free{\GF }X + \free{\GF }Y \ar[rr]^{[\free{\GF }\inl,\free{\GF }\inr]} & & \free{\GF }(X+Y) }.
\]
Since $\free{\GF }$ is a monad with canonical fixpoint solutions, it can be verified that so is $\free{\GF }(X+\Id)$. Moreover, $\pi$ is a monad morphism between $\GF X + \Id$ and $\free{\GF }(X+\Id)$. These observations allow us to prove the following.

\newcommand{\propfactEpsilonElim}{ For any $\GF $-system $e : X \to \GF X + X$ with internal transitions, consider the equation morphism $\pi_{X} \circ e: X \to \free{\GF }(X+X)$. Then:
\[\pi_0 \circ \epselim{e} = \cansol {(\pi_{X} \circ e)} : X \to \free{\GF }X.\]}
\begin{proposition}[Factorisation property of $e \mapsto \epselim{e}$] \label{prop:factEpsilonElim}
\propfactEpsilonElim
\end{proposition}
\begin{proof}
  This follows simply by an application of Proposition~\ref{prop:factorizationLemma} to $\epselim{e} = e^\ddagger$ and $\gamma = \pi$ with $Y = 0$. \qed
\end{proof}

\noindent We are now in position to show point {\bf \S 3}: soundness of $\epsilon$-elimination.
\newcommand{\propEpsilonElimSound}{For any $\GF $-system $e : X \to \GF X + X$ with internal transitions,
\[\bb{\epselim{e}} = \bb{e}.\]
}
\begin{theorem}[Eliminating internal transitions is sound] \label{th:epselimsound}
\propEpsilonElimSound
\end{theorem}
\begin{proof} The statement is shown by the following derivation.
\[
\begin{array}{lcl@{\qquad}l}
\bb{\epselim{e}}\ \ &= & \bb{ \altsol e} & \text{Definition of $\epselim{e}$}\\
    & = & \cansol {(\kappa_X \circ \altsol e)} & \text{Definition of $\bb{-}$ (Def. \ref{def:canfixSem})}\\
  & = &\cansol {(\pi_0 \circ \altsol e)} & \text{Definition of $\pi_0$} \\
& = & (\pi_{X} \circ e)^{\dagger\dagger} & \text{Proposition \ref{prop:factEpsilonElim}} \\
& = & \cansol {(\free{\GF }(\nabla_X) \circ (\pi_{X} \circ e))} & \text{double dagger law}\\
& = &  \cansol {\bar{e}} & \text{Definition of $\bar{e}$ (Def. \ref{def:reprHtofreeH}) and $\pi_X$} \\
& = &  \bb{e} & \text{Definition of $\bb{-}$.}
\end{array}
\]
\qed
\end{proof}
\section{Quotient Semantics}\label{ssec:quot}
When considering behavior of systems it is common to encounter spectrums of successively coarser equivalences. For instance, in basic process algebra trace equivalence can be obtained by quotienting bisimilarity with an axiom stating the distributivity of action prefixing by non-determinism~\cite{Rabinovich93}. There are many more examples of this phenomenon, including Mazurkiewicz traces, which we will describe below.

In this section we develop a variant of the canonical fixpoint semantics, where we can encompass in a uniform manner behaviors which are quotients of the canonical behaviors of the previous section (that is, the object $\free{\GF }0$).
\begin{assumption}\label{ass:quotient} Let $\C$, $\GF $, $\free{\GF }$ and $\GF X+\Id$ be as in Assumption \ref{ass:TheoryHSystems} and $\quotG \colon\free{\GF }\To Q$ a monad quotient for some monad $Q$. Moreover, suppose that for all $Y \in \C$ an initial $Q(\Id +Y)$-algebra exists.
\end{assumption}
Observe that, as Assumption \ref{ass:quotient} subsumes Assumption \ref{ass:TheoryHSystems}, we are within the framework of previous section, with the canonical fixpoint solution of $\free{\GF }$ providing semantics for $\free{\GF }$- and $\GF $-systems. For our extension, one is interested in $Q0$ as a semantic domain coarser than $\free{\GF }0$ and we aim at defining an interpretation for $\GF $-systems in $Q0$. To this aim, we first check that $Q$ has canonical fixpoint solutions.
\newcommand{\propQElgot}{Let $\C$, $\GF $, $Q$ and $\quotG : \free{\GF } \To Q$ be as in Assumption \ref{ass:quotient}. Then Assumption \ref{ass:CPPOenrichTcont} holds for $\C$ and $Q$, meaning that $Q$ is a monad with canonical fixpoint solutions (which satisfy the double dagger law by Corollary \ref{cor:cansolElgot}).}
\begin{proposition}\label{prop:QElgot}
\propQElgot
\end{proposition}
We use the notation $e \mapsto \quotsol e$ for the canonical fixpoint operator of $Q$. This allows us to define the semantics of $Q$-systems, analogously to what we did for $\free{\GF }$-systems in Definition \ref{def:canfixSem}. Moreover, the connecting monad morphism $\quotG \colon \free{\GF } \To Q$ yields an extension of this semantics to include also systems of transition type $\free{\GF }$ and $\GF $.
\begin{definition}[Quotient Semantics] \label{def:quotSem}
The quotient semantics of $\GF $-systems, with or without internal transitions, $\free{\GF }$-systems and $Q$-systems is defined as follows.
\begin{itemize}
  \item For a $Q$-system $e \colon X \to QX$, its semantics $\bbq{e} \colon X \to Q0$ is defined as $\quotsol{e}$ (note that $e$ can be regarded as an equation morphism for $Q$ with $Y=0$).
\item For an $\free{\GF }$-system $e\colon X \to \free{\GF }X$, its semantics $\bbq{e}\colon X \to Q0$ is defined as $\quotsol {(\quotG_X \circ e)}$.
\item For an $\GF $-system $e$---with or without internal transitions---its semantics $\bbq{e}\colon X \to Q0 $ is defined as $\quotsol {(\quotG_X \circ \bar{e})}$, where $\overline{e}$ is as in Definition \ref{def:reprHtofreeH}.
\end{itemize}
\end{definition}
The Factorisation Lemma (Proposition \ref{prop:factorizationLemma}) allows us to establish a link between the canonical fixpoint semantics $\bb{-}$ and the quotient semantics $\bbq{-}$.
\newcommand{\propfactQuotient}{ Let $e$ be either an $\free{\GF }$-system or an $\GF $-system (with or without internal transitions). Then:
\begin{equation} \label{eq:factQ} \bbq{e} = \quotG_0 \circ \bb{e}. \end{equation}}
\begin{proposition}[Factorisation for the quotient semantics] \label{prop:factQuotient}
\propfactQuotient
\end{proposition}
As a corollary we obtain that eliminating internal transitions is sound also for quotient semantics.
\newcommand{\corSoundQuot}{
For any $\GF $-system $e : X \to \GF X +X$ with internal transitions,
\[\bbq{e} =\bbq{\epselim e}.\]
}
\begin{corollary} \label{for:sound-quot}
\corSoundQuot
\end{corollary}
The quotient semantics can be formulated in a Kleisli category $\Kl(M)$ by further assuming $(c)$ below. This is needed to lift a quotient of monads from $\Sets$ to $\Kl(M)$.

\newcommand{\thquotcomparewithHJS}{Let $M \: \Sets \to \Sets$ be a monad and $H \: \Sets \to \Sets$ be an accessible functor satisfying the assumptions of Theorem \ref{thm:HJS}. By Proposition \ref{prop:liftingfreemonad} the free monad $\free{H}$ on $H$ lifts to a monad $\lift{\free{H}}\colon \Kl(M) \to \Kl(M)$ via a distributive law $\lambda \: \free{H}M \To M\free{H}$ with $\lift{\free{H}} = \free{\lift{H}}$. Let $\MM \: \Sets \to \Sets$ be a monad and $\quot \: \free{H} \To \MM$ a monad quotient such~that
\begin{itemize}
\item[(c)] for each set $X$, there is a map
 $\lambda_X'\colon \MM MX \to M\MM X$ making the following commute.
 $$\xymatrix{
 \free{H}MX \ar[d]_{{\quot}_{MX}} \ar[r]^{\lambda_X} & M\free{H}X \ar[d]^{M{\quot}_X}\\
 \MM MX \ar[r]_{\lambda'_X} & M\MM X
 }$$
 \end{itemize}
Then the following hold:
\begin{enumerate}
\item there is a monad $\lift{\MM}\colon \Kl(M) \to \Kl(M)$ lifting $\MM$ and a monad morphism $\lift{\quot} \colon \free{\lift{H}} \To \lift{\MM}$ defined as $\lift{{\quot}_X}=\J({\quot}_X)$; \label{pt:QuotKleisli2}
 \item $\Kl(M)$, $\lift{H}$, $\free{\lift{H}}$, $\lift{H}\J X + \Id$ (for a given set $X$), $\lift{\MM}$ and $\lift{\quot} \colon \free{\lift{H}} \To \lift{\MM}$ satisfy Assumption \ref{ass:quotient}. \label{pt:QuotKleisli3}
     \end{enumerate}}
\begin{theorem}\label{th:quotcomparewithHJS}
\thquotcomparewithHJS
\end{theorem}

Notice that condition~(c) and the first part of statement \ref{pt:QuotKleisli2}~are
related to~\cite[Theorem~1]{BHKR13}; however, that paper treats distributive
laws of monads over endofunctors.

\begin{example}[Mazurkiewicz traces] \label{Sec:MazurTraces}
This example, using a known equivalence in concurrency theory, illustrates the use of the quotient semantics developed in Section~\ref{ssec:quot}.

The trace semantics proposed by Mazurkiewicz~\cite{Mazurkiewicz77} accounts for concurrent actions. Intuitively, let $A$ be the action alphabet and $a,b\in A$. We will call $a$ and $b$ concurrent, and write $a\equiv b$, if the order in which these actions occur is not relevant. This means that we equate words that only differ in the order of these two actions, e.g.~$uabv$ and $ubav$ denote the same Mazurkiewicz trace.

To obtain the intended semantics of Mazurkiewicz traces we use the quotient semantics defined above\footnote{Mazurkiewicz traces were defined over labelled transition systems which are similar to NDA but where every state is final. For simplicity, we consider LTS here immediately as NDA.}. In particular, for Mazurkiewisz traces one considers a symmetric and irreflexive ``independence'' relation $I$ on the label set $A$. Let $\equiv$ be the least congruence relation on the free monoid $A^*$ such that
\[
(a,b) \in I \Rightarrow ab \equiv ba.
\]

We now have two monads on $\Sets$, namely $\free{H}X=A^* \times X + A^*$ and $\MM X=A^*\!/_\equiv \times X + A^*\!/_\equiv$. There is the canonical quotient of monads $\quot \: \free{H} \To \MM$ given by identifying words of the same $\equiv$-equivalence class. It can be checked that those data satisfy the assumptions of Theorem \ref{th:quotcomparewithHJS} and thus we are allowed to apply the quotient semantics $\bbq{-}$.
This will be given on an NDA $e\colon X \to \lift{H}X$ by first embedding it into $\free{\lift{H}}$ as $\bar e = \kappa_X \circ e \colon X\to \free{\lift{H}}X$ and then into $\lift{\MM}$ as $\lift{\quot}_X \circ \bar e \colon X\to \lift{\MM}X$.
To this morphism we apply the canonical fixpoint operator of $\lift{\MM}$ to obtain $\quotsol{(\lift{\quot}_X \circ \bar e)}$, that is, the semantics $\bbq{e} \colon X \to R0=A^{*}\!/\!\equiv$.
It is easy to see that this definition captures the intended semantics: for all states $x\in X$
\[
\bbq{e}(x) = \{ [w]_\equiv \mid w \in \bb{e}(x)\}.
\]
Indeed, by Proposition \ref{prop:factQuotient}, $\bbq{e}= \lift{\quot}_0 \circ \bb{e}$ and $\lift{\quot}_0 \colon  \free{\lift{H}}0 \to \lift{\MM}0$ is just $\J\quot_0$ where $\quot_0 \colon A^* \to A^*\!/_\equiv$ maps every word $w$ into its equivalence class $ [w]_\equiv$.
\end{example}

\section{Discussion}\label{Sec:Discussion}
The framework introduced in this paper provides a uniform way to express the semantics of systems with internal behaviour via canonical fixpoint solutions. Moreover, these solutions are exploited to eliminate internal transitions in a sound way, i.e., preserving the semantics. We have shown our approach at work on NDA with $\epsilon$-transitions but, by virtue of Theorem \ref{th:comparewithHJS}, it also covers all the examples in~\cite{HasuoJS:07} (like probabilistic systems) and more (like the weighted automata on positive reals of~\cite{SW13}).

It is worth noticing that, in principle, our framework is applicable also to examples that do not arise from Kleisli categories. Indeed the theory of Section~\ref{sec:Theory} is formulated 
for a general category $\cat{C}$: Assumption \ref{ass:TheoryHSystems} only requires $\cat{C}$ to be $\cppo$-enriched and the monad $T$ to be locally continuous. The role of these assumptions is two-fold: (a) ensuring the initial algebra-final coalgebra coincidence and (b) guaranteeing that the canonical fixpoint operator $e \mapsto \cansol{e}$ satisfies the \emph{double dagger law}. If (a) implies (b), we could have formulated our theory just assuming the coincidence of initial algebra and final coalgebra and without any $\cppo$-enrichment. Condition~(a) holds for some interesting examples not based on Kleisli categories, e.g. for examples in the category of join semi-lattices. Therefore it is of relevance to investigate the following question: given a monad $T$ with initial algebra-final coalgebra coincidence, under which conditions does the canonical fixpoint solution provided by $T$ satisfy the double dagger law?

As a concluding remark, let us recall that our original question concerned the problem of modeling the semantics of systems where labels carry an algebraic structure.
In this paper we have mostly been focusing on automata theory, but there are many other examples in which the information carried by the labels has relevance for the semantics of the systems under consideration:
in logic programming labels are substitutions of terms; in (concurrent) constraint programming they are elements of a lattice; in process calculi they are actions representing syntactical contexts and in tile systems~\cite{DBLP:conf/birthday/GadducciM00} they are morphisms in a category. We believe that our approach provides various insights towards a coalgebraic semantics for these computational models.

\paragraph{Acknowledgments.}
We are grateful to the anonymous referees for valuable comments.
The work of Alexandra Silva is partially funded by the ERDF through the
Programme COMPETE and by the Portuguese Foundation for Science and
Technology, project ref.~\texttt{FCOMP-01-0124-FEDER-020537} and
\texttt{SFRH/BPD/}\texttt{71956/2010}.
The first and the fourth author acknowledge support by project \texttt{ANR 12IS0} \texttt{2001 PACE}.

\bibliographystyle{splncs03}

\newpage
\appendix
\appendix

\section{Proofs of Section~\ref{Sec:Trace}}
%
In this appendix, we show the proofs of Proposition~\ref{prop:liftingfreemonad} and \ref{lem:Hstar}.
The proofs of the other results shown in Section~\ref{Sec:Trace} can be found in the referred literature.

\begin{proposition_for}{prop:liftingfreemonad}
\propliftingfreemonad
\end{proposition_for}
\begin{proof}

Let $\lambda \colon HM \to MH$ be the distributive law of the functor $H$ over the monad $M$ corresponding to the lifting $\lift{H}$ (see Proposition~\ref{LiftProp}).
For an object $X$, we define $\gamma_X \colon \free{H}M \to \free{H}M$ by the universal property of the initial $H(-)+MX$-algebra $\free{H}(MX)$.
\begin{equation}
  \vcenter{
  \xymatrix{
    H\free{H}MX \ar[rr]^-{\tau_{MX}} \ar[d]_-{H\gamma_X} && \free{H}MX \ar@{.>}[d]^{\gamma_X} & MX \ar[l]_-{\eta_{MX}} \ar[ld]^{M\eta_X}\\
    HM\free{H}X \ar[r]_-{\lambda_{TX}} & MH\free{H}X \ar[r]_-{M \tau_X}
    & M\free{H}X
  }}
\end{equation}
By diagram chasing, one can show that $\gamma \colon \free{H}M \To M\free{H}$ is a distributive law of the monad $\free{H}$
over the monad $M$ and, by Proposition~\ref{LiftProp}, we have a lifting $\lift{\free{H}}\colon \Kl(M)\to \Kl(M)$.

For proving $\lift{\free{H}} = \free{\lift{H}}$,
take $\alpha_X \colon H(\free{H}(X))+X \to \free{H}(X)$ to be the initial $H(-)+X$-algebra and observe that $\J(\alpha)$ is the initial
$\lift{H}(-)+X$-algebra (Proposition~\ref{prop:liftinginitialalgebra}). The fact that the units and the multiplications of $\lift{\free{H}}$ and $\free{\lift{H}}$ coincide is immediately proved by functoriality of $\J$.
%
 \qed
\end{proof}

\begin{proposition_for}{lem:Hstar}
 \lemHstar
\end{proposition_for}
\begin{proof}
  First recall that $\free{H} X$ is a free $H$-algebra with the
  structure $\tau_X$ and the universal morphism $\eta_X$
  (cf.~Example~\ref{ex:mnds}(5)). Equivalently, $\alpha_X = [\tau_X,
  \eta_X]: H(\free{H} X) + X \to \free{H} X$ is an initial algebra for
  $H(-) +X$. Given a morphism $f: X \to Y$, $\free{H} f$ is defined by
  initiality; more precisely, $\free{H} f$ is the unique morphism such
  that the following diagram commutes:
  \[
  \xymatrix@C+1pc{
    H(\free H X) + X \ar[rr]^-{\alpha_X}
    \ar[d]_{H(\free H f) + \id}
    &&
    \free H X
    \ar@{-->}[d]^{\free H f}
    \\
    H(\free H Y) + X
    \ar[r]_-{\id +f }
    &
    H(\free H Y) + Y
    \ar[r]_-{\alpha_Y}
    &
    \free H Y
    }
  \]
  Now recall that $\alpha_X$ is an isomorphism and consider the following function
  \[
  \Phi: \cat C(X,Y) \times \cat C(\free H X, \free H Y) \to \cat
  C(\free H X, \free H Y)
  \]
  with
  \[
  \Phi(f,h) = \alpha_Y \cdot (Hh + f) \cdot \alpha_X^{-1}.
  \]
  Since $H$ is locally continuous, we see that $\Phi$ is continuous (in
  both arguments). Clearly, $\free H f$ is the unique fixpoint of $\Phi(f, -)$. To see that $\free H f$ is
  locally continuous let $f_n: X \to Y$ be an $\omega$-chain in $\cat
  C(X,Y)$. It is easy to see that $\bigsqcup_{n < \omega} \free H f_n$
  is a fixpoint of $\Phi\left(\bigsqcup_{n< \omega} f_n,-\right)$; indeed we
  have (using continuity of $\Phi$):
  \begin{eqnarray*}
    \bigsqcup\limits_{n < \omega} \free H f_n
    & = & \bigsqcup\limits_{n < \omega} \Phi(f_n, \free H f_n) \\
    & = & \Phi\left(\bigsqcup\limits_{n < \omega} f_n,
      \bigsqcup\limits_{n < \omega}  \free H f_n\right).
  \end{eqnarray*}
  Thus, by the uniqueness of the fixpoint $\free
  H\left(\bigsqcup_{n < \omega} f_n\right)$ we have
  \[
  \free H\left(\bigsqcup\limits_{n < \omega} f_n \right) =
  \bigsqcup\limits_{n < \omega} \free H f_n
  \]
  as desired. \qed
\end{proof}

Finally, we record a simple lemma for future use:

\begin{proposition}
  \label{lem:quot}
 Let $H'$ be a quotient functor of the locally continuous functor $H$ on the \cppo-enriched category $\C$. Then $H'$ is locally continuous, too.
\end{proposition}
\begin{proof}
 Suppose that $\quotG \colon H \to H'$ is an epi natural transformation.
 Consider an $\omega$-chain $(f_n)_{n<\omega}$ in $\C(X,Y)$. To prove that $H'(\bigsqcup f_n) = \bigsqcup H'f_n$ we show that
\begin{align*}
H' (\bigsqcup_{n< \omega} f_n) \circ q_X &= & q_Y \circ H(\bigsqcup_{n< \omega} f_n) \tag{naturality of $\quotG$}\\
  & = &q_Y \circ (\bigsqcup_{i< \omega} Hf_n) & \tag{$H$ locally continuous} \\
& = & \bigsqcup_{n< \omega} (q_Y \circ  Hf_n ) & \tag{continuity of comp.} \\
& = & \bigsqcup_{n< \omega} (H'f_n \circ q_X ) & \tag{naturality of $\quotG$} \\
& = &  (\bigsqcup_{n< \omega} H'f_n) \circ q_X  & \tag{continuity of comp.} \\
\end{align*}
and we use that $\quotG_X$ is epi.
\qed
\end{proof}


\section{Proofs of Section~\ref{sec:Theory}}\label{App:proofsTheory}

In this appendix, we report the proofs of the results stated in Section~\ref{sec:Theory}, apart from Theorem~\ref{th:comparewithHJS} that we prove separately in the next appendix.

\begin{proposition_for}{prop:lfp=cfpsolution}
\proplfpcfpsolution
\end{proposition_for}
\begin{proof} It suffices to show that $\fromInit \circ \toFinal{e}$ is the least fixpoint of the continuous function $\Phi_e$ on $\C(X,TY)$, defined as in Example~\ref{ex:LfpSolCPO}. To this aim, first observe that the least fixpoint of $\Phi_e$ can be obtained as the $\omega$-join
\[\lsol e = \bigsqcup_{n< \omega}\lsol e_n : X \to TY\]
where $\lsol e_0 = \bot_{X,TY}$ and $\lsol e_{n+1} = \mu^T_Y \circ T[\lsol e_{n},\eta_Y^T] \circ e$.

An analogous observation can be made for the coalgebra morphism $\toFinal{e} : X \to \carrier_Y$. By finality of $\carrier_Y$, $\toFinal{e}$ is the unique map making the left-hand square in \eqref{diag:canonicalSol} commute. In particular, it is the \emph{least} function---in the cpo $\C(X, \carrier_Y)$---to do so: thus it is the least fixpoint of a continuous function, expressed by the $\omega$-join
\[\toFinal{e} = \bigsqcup_{n< \omega} c_n : X \to \carrier_Y\]
 where $c_0 = \bot_{X,\carrier_Y}$ and $c_{n+1} = \iota_Y \circ T(c_n + \id_Y) \circ e$.
 Analogously, by initiality of $\carrier_Y$, $\fromInit : \carrier_Y \to TY$ is the unique---and thus the least---fixpoint of a continuous function on $\C(\carrier_Y,TY)$, calculated as follows:
\[\fromInit = \bigsqcup_{n< \omega} d_n : \carrier_Y \to TY\]
 where $d_0 = \bot_{\carrier_Y,TY}$ and $d_{n+1} = \mu^T_Y \circ T[d_{n},\eta_Y^T] \circ \iota^{-1}_Y$.

\noindent We now show by induction on $n$ that $\lsol e_n = d_n \circ c_n$:
\begin{itemize}
  \item for $n=0$, by left-strictness of composition we have
      \[\lsol e_0 = \bot_{X,TY} = \bot_{\carrier_Y,TY} \circ \bot_{X,\carrier_Y} = d_0 \circ c_0.\]
  \item For the inductive step, consider the following derivation:
      \begin{align*}d_{n+1} \circ c_{n+1} =\ &  \mu^T_Y \circ T[d_n,\eta^T_Y] \circ \iota^{-1}_Y \circ \iota_Y \circ T(c_n + \id_Y) \circ e \\
       =\ & \mu^T_Y \circ T[d_n,\eta^T_Y] \circ T(c_n + \id_Y) \circ e \\
      =\ & \mu^T_Y \circ T[d_n\circ c_n,\eta^T_Y] \circ e\\
        \overset{IH}{=}\ & \mu^T_Y \circ T[\lsol e_n,\eta^T_Y] \circ e \\
        =\ &\lsol e_{n+1}.
      \end{align*}
\end{itemize}

\noindent Thus we are allowed to conclude:
\[\lsol e = \bigsqcup_{n< \omega}\lsol e_n = \bigsqcup_{n< \omega} d_n \circ c_n = \bigsqcup_{n< \omega} d_n \circ \bigsqcup_{n< \omega} c_n = \fromInit \circ \toFinal{e},\]
where the third equality is given by continuity of composition in $\C$.\qed
\end{proof}

\begin{remark}\label{rm:uniqueness} In the proof of Proposition~\ref{prop:lfp=cfpsolution} one observes that both $\toFinal{e}$ and $\fromInit$ are \emph{unique} fixpoints for the continuous functions on $\cat C(X, \carrier_Y)$ and $\cat C(\carrier_Y, TY)$, respectively, corresponding to commutativity of the two inner squares in \eqref{diag:canonicalSol}. Nonetheless, the same is not true for their composite $\cansol e$, which we just prove to be the \emph{least} solution for $e$: there are possibly other maps making the outer rectangle in \eqref{diag:canonicalSol} commute (cf. Remark~\ref{rm:multiple}).
\end{remark}

 \begin{proposition_for}{prop:factorizationLemma}
 \propfact
 \end{proposition_for}
\begin{proof} First we construct the canonical fixpoint solution for $e \: X \to T(X+Y)$ and $e' \stackrel{\mathrm{def}}= \gamma_{X+Y} \circ e \colon X \to T'(X+Y)$. The former will factor through the initial $T(\Id +Y)$-algebra $\iota_Y : T(\carrier_Y+ Y) \to \carrier_Y$ and the latter through the initial $T'(\Id +Y)$-algebra ${\iota'}_Y \: T'(\carrier'_Y+ Y) \to \carrier'_Y$ as in the diagram:
\begin{equation*}\label{diag:factcanonicalSol}\small
\vcenter{
    \xymatrix@R=12pt@C=7pt{
    X \ar@{-->}[rr]^{\toFinal{e}} \ar[dd]_{e} \ar@{=}[dddr] & & \carrier_Y \ar@/_/[dd]_{\iota_Y^{-1}} \ar@{-->}[rr]^{\fromInit} & & TY \ar[dddr]^{\gamma_Y} \\
     & & & & TTY \ar[u]^{\mu^T_Y} \\
    T(X+Y)  \ar@{-->}[rr]_>>>>>>>>>{T(\toFinal{e} + \id_Y)} & & T(\carrier_Y +Y)
    \ar@/_/[uu]_{\iota_Y}
    \ar@{-->}[rr]_{T( \fromInit + \id_Y)} & & T(TY+Y) \ar[u]^{T[\id_{T Y},\eta^T_Y]}
    \\
     &   X \ar@{-->}[rr]^{\toFinal{e'}} \ar[dd]_{e'} & & \carrier'_Y \ar@/_/[dd]_{{\iota'}_Y^{-1}} \ar@{-->}[rr]^{\fromInit'} & & T'Y\\
     & & & & & T'T'Y \ar[u]_{\mu^{T'}_Y} \\
    & T'(X+Y)  \ar@{-->}[rr]_{T'(\toFinal{e'} + \id_Y)} & & T'(\carrier'_Y +Y)
    \ar@/_/[uu]_{{\iota'}_Y}
    \ar@{-->}[rr]_{T'( \fromInit' + \id_Y)} & & T'(T'Y+Y) \ar[u]_{T'[\id_{T' Y},\eta^{T'}_Y]}
    }
}
\end{equation*}
The statement of the Proposition amounts to show that the top face of the diagram commutes, that is,
\begin{equation}\label{eq:fact}
\gamma_Y \circ \fromInit \circ \toFinal{e} = \fromInit' \circ \toFinal{e'}.
\end{equation}
We are going to prove \eqref{eq:fact} by exploiting the initiality of $\carrier_Y$ and finality of $\carrier'_Y$. For this purpose, it is convenient to make the following observation:
\begin{itemize}
  \item[$(*)$] any $T'(\Id +Y)$-algebra $f \colon T'(A+Y) \to A$ canonically induces a $T(\Id +Y)$-algebra  $f \circ \gamma_{A+Y} \colon T(A+Y) \to A$ and the same---by naturality of $\gamma$---for algebra homomorphisms. Dually, any $T(\Id +Y)$-coalgebra $g \colon B \to T(B+Y)$ canonically induces a $T'(\Id +Y)$-algebra $\gamma_{B+Y} \circ g : T'(B+Y) \to B$ and the same for coalgebra homomorphisms.
      \begin{equation*}
    \xymatrix@C=40pt@R=15pt{
    A \ar[r]^{h} & \tilde{A} \\
    T'(A+Y) \ar[u]^{f} \ar[r]^{T'(h+\id_Y)} & T'(\tilde{A}+Y) \ar[u]^{\tilde{f}} \\
    T(A+Y) \ar[u]^{\gamma_{A+Y}} \ar[r]^{T(h+\id_Y)} & T(\tilde{A}+Y) \ar[u]^{\gamma_{\tilde{A}+Y}}
    }
    \qquad
    \xymatrix@C=40pt@R=15pt{
    B \ar[r]^{u} \ar[d]^{g}& \tilde{B} \ar[d]^{\tilde{g}}\\
    T'(B+Y) \ar[d]^{\gamma_{B+Y}} \ar[r]^{T'(u+\id_Y)} & T'(\tilde{B}+Y) \ar[d]^{\gamma_{\tilde{B}+Y}} \\
    T(B+Y) \ar[r]^{T(u+\id_Y)} & T(\tilde{B}+Y)
    }
\end{equation*}
\end{itemize}
By observation $(*)$, $\carrier'_Y$ has a $T(\Id +Y)$-algebra structure and thus by initiality there is a unique $T(\Id +Y)$-algebra morphism $a : \carrier_Y \to \carrier'_Y$. Then our claim \eqref{eq:fact} reduces to the commutativity of the following diagram.
      \begin{equation} \label{diag:fact12}
    \vcenter{
        \xymatrix@C=80pt@R=30pt{
        X \ar[r]^{\toFinal{e}} \ar[dr]_{\toFinal{e'}} & \carrier_Y \ar@{}[ld]|<<<<<<<<<<{\circone{2}} \ar@{}[rd]|{\circone{1}} \ar[r]^{\fromInit} \ar@{.>}[d]^-{a}& TY \ar[d]^{\gamma_Y} \\
        & \carrier'_Y\ar[r]^{\fromInit'} & T'Y
        }
    }
\end{equation}
We address commutativity of $\circone{1}$ and of $\circone{2}$ separately.
\begin{itemize}
  \item[-$\circone{1}$-] By observation $(*)$, $T'Y$ has also a $T(\Id +Y)$-algebra structure. Then, by initiality of $\carrier_Y$, for commutativity of $\circone{1}$ it suffices to show that $\gamma_Y \circ \fromInit$ and $\fromInit' \circ a$ are $T(\Id +Y)$-algebra morphisms. For this purpose, first observe that by construction $\fromInit$ and $a$ are $T(\Id +Y)$-algebra morphism and the same for $\fromInit'$ in virtue of observation $(*)$. Hence it suffices to prove that also $\gamma_Y$ is a $T(\Id +Y)$-algebra morphism. That is given by commutativity of the following diagram
        \begin{equation*}\label{diag:gammaAlgMorf}
    \vcenter{
        \xymatrix@C=60pt@R=30pt{
        TY \ar[rr]^{\gamma_Y} \ar@{}[rrd]|{\circone{X}} & & T'Y \\
        TTY \ar[u]^{\mu^T_Y} \ar[r]^{T\gamma_Y} \ar@{}[rd]|{\circone{Y}}  & TT'Y \ar[r]^{\gamma_{T'Y}} \ar@{}[rd]|{\circone{Z}} & T'T'Y \ar[u]_{\mu^{T'}_Y} \\
        T(TY+Y) \ar[r]^{T(\gamma_Y + \id_Y)} \ar[u]^{T[\id_{T'Y} + \eta^{T}_Y]} & T(T'Y+Y) \ar[r]^{\gamma_{T'Y+Y}} \ar[u]|{T[\id_{T'Y} + \eta^{T'}_Y]} & T'(T'Y+Y) \ar[u]_{T'[\id_{T'Y} + \eta^{T'}_Y]}
        }
    }
\end{equation*}
      where $\circone{X}$ and $\circone{Y}$ commute because $\gamma$ is a monad morphism and $\circone{Z}$ by naturality of $\gamma$.
  \item[-$\circone{2}$-] To show that also $\circone{2}$ in \eqref{diag:fact12} commutes, we first check that $a : \carrier_Y\to \carrier'_Y$ is also a $T'(\Id +Y)$-coalgebra morphism:
        \begin{equation*}
    \vcenter{
        \xymatrix@C=40pt@R=15pt{
            \carrier_Y\ar[rrr]^{a} \ar@/_/[d]_{\iota_Y^{-1}} &&& \carrier'_Y \ar@/_/[d]_{{\iota'}_Y^{-1}} \\
            T(\carrier_Y+Y) \ar@/_/[u]_{\iota_Y} \ar[drrr]_>>>>>>>>>>>>>>{T(a+\id_Y)} \ar[d]^{\gamma_{\carrier_Y+ Y}} \ar@{}[rd]|{\circone{V}} &&& T'(\carrier'_Y +Y) \ar@/_/[u]_{{\iota'}_Y} \ar@{}[ld]|{\circone{U}} \\
            T'(\carrier_Y +Y) \ar[urrr]^>>>>>>>>>>>>>{T'(a+\id_Y)} &&& T(\carrier'_Y+Y) \ar[u]_{\gamma_{\carrier'_Y+ Y}}
            }
    }
\end{equation*}
      In the diagram above, the pentagon with angle $\circone{U}$ commutes because $a$ is a $T(\Id +Y)$-algebra morphism, whereas the pentagon with angle $\circone{V}$ commutes by naturality of $\gamma$ applied to the maps $T(a + \id_Y)$ and $T'(a + \id_Y)$.

      To conclude, observe that $\toFinal{e'}$ (by construction) and $\toFinal{e}$ (by observation $(*)$) are also $T'(\Id +Y)$-coalgebra morphisms. Thus $a \circ \toFinal{e} = \toFinal{e'}$ by finality of $\carrier'_Y$, meaning that $\circone{2}$ in \eqref{diag:fact12} commutes.
\end{itemize}\qed
\end{proof}

 \begin{proposition_for}{prop:ass2->ass1}
 \propass
 \end{proposition_for}
 \begin{proof} We check that the two monads satisfy Assumption~\ref{ass:CPPOenrichTcont}. For all $Y\in\C$, the condition on the existence of initial algebras for the endofunctors $\free{\GF }(\Id+Y)$ and $\GF X + \Id +Y$ is already guaranteed by Assumption~\ref{ass:TheoryHSystems}. It remains to show local continuity. As $\GF $ is locally continuous and all free $\GF $-algebras exist, the monad $\free{\GF }$ is also locally continuous by Proposition~\ref{lem:Hstar}. Local continuity of $\GF X + \id$ is immediate by the fact that all copairing maps $[-,-]:\cat C(Y,Z) \times \cat C(Y', Z) \to \cat C(Y+Y',Z)$ in the $\cppo$-enriched category $\C$ are continuous (\emph{cf.} Section \ref{sec:cppo}).
 \qed
 \end{proof}

 \begin{proposition_for}[Factorisation property of $e \mapsto \epselim{e}$]{prop:factEpsilonElim}
\propfactEpsilonElim
\end{proposition_for}
\begin{proof}
Let us use the notation $e \mapsto \sol e$ for the canonical fixpoint solution operator of $\free{\GF }(X+\Id)$. We now apply Proposition~\ref{prop:factorizationLemma} to show that solutions of $\free{\GF }(X+\Id)$ factorize through the ones of $\GF X + \Id$. The connecting monad morphism is $\pi : \GF X + \Id \to \free{\GF }(X+\Id)$, defined above. Proposition~\ref{prop:factorizationLemma} yields the following factorisation property:
\begin{itemize}
\item[$(*)$] for any $Y,Z \in C$ and equation morphism $e\colon Z \to \GF X + Z + Y$, consider $\pi_{Z + Y} \circ e \colon Z \to \free{\GF }(X + Z +Y)$. The solution $\sol {(\pi_{Z + Y} \circ e)} \colon Z \to \free{\GF }(X + Y)$ provided by $\free{\GF }(X+\Id)$ factorises as $\pi_Y \circ \altsol e$, where $\altsol e \colon Z \to \GF X +Y$ is the solution provided by $\GF X + \id$ to $e$.
\end{itemize}
If we fix $Z = X$ and $Y = 0$, then $(*)$ says: for any $\GF $-system $e: X \to \GF X + X$ with internal computation, consider the equation morphism $(\pi_{X + 0} \circ e : X) \to \free{\GF }(X + X + 0)$ for $\free{\GF }(X+\Id)$ with parameter $Y = 0$. Then the following diagram commutes:
\begin{equation}\label{diag:factHintoHstar}
\xymatrix{
    X \ar[rr]^{\sol {(\pi_{X} \circ e)}} \ar[drr]_{\altsol e} && \free{\GF }X \\
    && \GF X \ar[u]_{\pi_0}
}
\end{equation}
To conclude our argument, we observe that the the system $\pi_{X + 0} \circ e \colon X \to \free{\GF }(X + X + 0)$ can be also seen as an equation for $\free{\GF }$ with parameter $Y = X + 0$. This means that also $\free{\GF }$ provides a solution to such equation, which can be checked to coincide with the one given by $\free{\GF }(X+\Id)$, that is, $\sol {(\pi_{X} \circ e)} = \cansol {(\pi_{X} \circ e)}$. Then the main statement is proven by be the following derivation:
\begin{align*}
\pi_0 \circ \epselim{e}\ \ &= & \pi_0 \circ \altsol e \tag{Definition of $\epselim{e}$}\\
    & = & \sol {(\pi_{X} \circ e)} \tag{commutativity of \eqref{diag:factHintoHstar}}\\
  & = &\cansol {(\pi_{X} \circ e)}. & \tag{observation above}
\end{align*} \qed
\end{proof}

\section{Proof of Theorem~\ref{th:comparewithHJS}}\label{app:proofthComparewithHJS}

This section is devoted to prove Theorem~\ref{th:comparewithHJS}. To this aim, we first give more details on accessible endofunctors and how they yield a canonical free algebra construction.

\begin{remark}
  \label{rem:chain}
  \begin{enumerate}[(1)]
  \item \label{pt:bounded} Ad\'amek and Porst~\cite{ap:04} showed that an endofunctor $H$ on $\Sets$ is
    accessible iff is it bounded in the following sense: there exists
    a cardinal $\lambda$ such that for every set $A$, every element of
    $HA$ lies in the image of $Hb$ for some $b: B \hookrightarrow A$
    of less than $\lambda$ elements.

  \item \label{pt:HaccessibleFreeAlg} Recall from~\cite{adamek:74} that for an accessible endofunctor
    $H$ on a cocomplete category $\C$ (not only the initial but) all
    \emph{free} $H$-algebras exist and are obtained from an inductive
    construction. More precisely, for every object $X$ of $\C$ define
    the following ordinal indexed \emph{free-algebra-chain}:
    \begin{eqnarray*}
      H_0 X & = & X, \\
      H_{i+1} X & = & HH_i X + X, \\
      H_j X & = & \colim\limits_{i < j} H_i X\qquad \text{for a
        limit ordinal $j$}.
    \end{eqnarray*}
    Its connecting morphisms $u_{i,j}: H_i X \to H_j X$ are
    uniquely determined by
    \[
    \begin{array}{rcl}
      u_{0,1} & = & (\xymatrix@1{X \ar[r]^-\inr & HX + X}),\\
      u_{i+1,j+1} & = & (\xymatrix@1{
        HH_i X + X \ar[rr]^-{Hu_{i,j} + X} && HH_j X + X}), \\
      \multicolumn{3}{p{10cm}}{$u_{i,j}$ ($i < j$) is the colimit cocone for
        limit ordinals $j$.}
    \end{array}
    \]
    Indeed, this defines an ordinal indexed chain uniquely (up to
    isomorphism). The ``missing'' connecting maps are determined by
    the universal property of colimits, e.g.~$u_{\omega,\omega+1}$ is
    unique such that $u_{\omega,\omega+1} \cdot u_{i+1,\omega} =
    u_{i+1,\omega+1} = Hu_{i,\omega}$ for all $i < \omega$.

    Now suppose that $H$ preserves $\lambda$-filtered
    colimits. Then $u_{\lambda, \lambda+1}$ is an isomorphism and one
    can show that $H_\lambda X$ is a free $H$-algebra on $X$ with the
    structure and universal morphism given by
    $u_{\lambda,\lambda+1}^{-1}$.
  \item \label{pt:freeHaccessible} As we saw previously, the assignment of a free $H$-algebra on $X$
    to any object $X$ yields a free monad on $H$; thus, in item~(2)
    above we have $\free{H} = H_\lambda$. Now notice that the construction
    in the previous point can be written object free; we obtain $\free{H}$
    after $\lambda$ steps of the following chain in the category of
    endofunctors on $\C$:
    \begin{eqnarray*}
      H_0 & = & \Id, \\
      H_{i+1} & = & HH_i + \Id, \\
      H_j & = & \colim\limits_{i < j} H_j \qquad\textrm{for limit
        ordinals $i$.}
    \end{eqnarray*}
    The connecting natural transformations $H_i \To H_j$ have the
    components described as connecting morphisms in item~(2).

    As a consequence we see that if $H$ is accessible then so is
    $\free{H}$; indeed, all $H_i$  preserve $\lambda$-filtered colimits if
    $H$ does.
  \end{enumerate}
\end{remark}

The next Proposition is instrumental in relating accessiblity of an endofunctor with the existence of initial algebras for its lifting.

\newcommand{\propliftingCoproduct}{Let $\C$ a cocomplete category, $M \colon \C \to \C$ be a monad and $\GFG \: \C \to \C$ be an accessible endofunctor with a lifting $\lift{\GFG} \colon \Kl(M) \to \Kl(M)$. Then for all $X \in \Kl(M)$ both the initial $\lift{\GFG}(\Id+X)$-algebra and the initial $\lift{\GFG}(\Id)+X$-algebra exist.}
\begin{proposition} \label{prop:liftingCoproduct}
\propliftingCoproduct
\end{proposition}
\begin{proof} As the left adjoint $\J \: \C \to \Kl(M)$ is defined as the identity on objects, without loss of generality we can prove our statement for an object $\J Y \in \Kl(M)$, where $Y \in \C$.

First we observe that the endofunctor $Y + \Id \colon \C \to \C$ (\emph{cf.} Example~\ref{ex:mnds}.\ref{pt:exceptionmonad}) always has a lifting to $\Kl(M)$. Indeed, because the left adjoint $\J \: \C \to \Kl(M)$ preserves coproducts, we have
 \[\J \circ ( \Id + Y) = \J(\Id) + \J Y = (\Id +\J Y) \circ \J\]
 implying that $\Id + \J Y \: \Kl(M)\to \Kl(M)$ is a lifting of $\Id +Y \colon \C \to \C$.

 Now we can compose the $\C$-endofunctors $\GFG$ and $\Id + Y $ in two different ways, obtaining $ \GFG(\Id)+ Y \colon \C \to \C$ and $\GFG(\Id+Y)\colon \C \to \C$. It is straightforward to check that the composite of two liftings is a lifting of the composite functor. This means that we have liftings $\lift{\GFG}(\Id) + \J Y\colon \Kl(M) \to \Kl(M)$ and $\lift{\GFG}(\Id+\J Y)\colon \Kl(M) \to \Kl(M)$ respectively of $\GFG(\Id) + Y\colon \C \to \C$ and $\GFG(\Id+Y)\colon \C \to \C$.

 The next step is to use accessibility to get initial algebras in $\C$ that will be then lifted to $\Kl(M)$. To this aim, we observe that both functors $\GFG(\Id)+ Y \colon \C \to \C$ and $\GFG(\Id+Y)\colon \C \to \C$ are accessible, because the functor $Y + \Id$ is clearly accessible and $\GFG$ is assumed to have this property.

Thus as observed in Remark~\ref{rem:chain}.\ref{pt:HaccessibleFreeAlg} both an initial $\GFG(\Id)+ Y$-algebra and an initial $\GFG(\Id+Y)$-algebra exist. Then Proposition~\ref{prop:liftinginitialalgebra} yields the existence both of an initial $\lift{\GFG}(\Id)+ \J Y$-algebra and an initial $\lift{\GFG}(\Id+\J Y)$-algebra. \qed
\end{proof}

We are now ready to supply a proof of Theorem \ref{th:comparewithHJS}.

\begin{theorem_for}{th:comparewithHJS}
\thcomparewithHJS
\end{theorem_for}
\begin{proof}
Since $\Kl(M)$ inherits coproducts from $\Sets$, we only need to check the following properties:
\begin{enumerate}
  \item all free $\lift{H}$-algebras exist;
  \item for all $Y \in \Kl(M)$, the initial $\free{\lift{H}}(\Id+Y)$-algebra exists;
  \item for all $Y \in \Kl(M)$, the initial $\lift{H}\J X + \Id + Y$-algebra exists.
\end{enumerate}
In virtue of Proposition~\ref{prop:liftingCoproduct}, the three properties are implied respectively by the following statements:
\begin{enumerate}
  \item the functor $H\:\Sets \to \Sets$ is accessible;
  \item the functor $\free{\lift{H}} \: \Kl(M) \to\Kl(M)$ is the lifting of $\free{H} \: \Sets \to \Sets$ and $\free{H}$ is accessible;
  \item the functor $\lift{H}\J X + \Id \: \Kl(M) \to\Kl(M)$ is the lifting of $H X + \Id \: \Sets \to \Sets$ and $H X + \Id$ is accessible.
\end{enumerate}
The first point is given by assumption. For the second point, $\free{H}$ is accessible by Remark~\ref{rem:chain}.\ref{pt:freeHaccessible} and $\free{\lift{H}} \: \Kl(M) \to\Kl(M)$ is its lifting by Proposition~\ref{prop:liftingfreemonad}. For the third point, since the identity $\Id \: \Sets\to \Sets$ and the constant functor $H X \: \Sets\to\Sets$ are clearly accessible and coproducts preserve this property, then $H X + \id \: \Sets\to\Sets$ is also accessible. As the left adjoint $\J \: \Sets\to\Kl(M)$ preserves coproducts, it is immediate to check that $\lift{H}\J X + \id \: \Kl(M) \to\Kl(M)$ is the lifting of $H X + \Id \: \Sets\to\Sets$. Indeed:
\[\J \circ (H X + \Id) = \J H X + \J(\Id) = \lift{H} \J X + \J(\Id) =  (\lift{H} \J X + \Id) \circ \J. \]
This concludes the proof of the three properties above. \qed
\end{proof}

\section{Proofs of Section \ref{ssec:quot}}

In this appendix, we report the proofs of the results stated in Section~\ref{ssec:quot}, apart from Theorem~\ref{th:quotcomparewithHJS} that we prove separately in the next appendix.

\begin{proposition_for}{prop:QElgot}
\propQElgot
\end{proposition_for}
\begin{proof} We need to check the following:
 \begin{enumerate}
   \item for all $Y \in \C$ an initial $Q(\Id+Y)$-algebra exist and
   \item $Q$ is locally continuous.
 \end{enumerate}
 The first point is given by Assumption~\ref{ass:quotient}. For the second point, we already checked with Proposition~\ref{prop:ass2->ass1} that our assumptions on $\C$ and $\GF $ imply that $\free{\GF }$ is locally continuous. Then, by Proposition~\ref{lem:quot}, $Q$ has the same property. \qed
\end{proof}

\begin{proposition_for}[Factorisation for the quotient semantics]{prop:factQuotient}
\propfactQuotient
\end{proposition_for}
\begin{proof} We instantiate the statement of Proposition~\ref{prop:factorizationLemma} to the monads $\free{\GF }$, $Q$ and the monad morphism $\quotG \colon \free{\GF } \To Q$. It amounts to commutativity of the following diagram for a given $\free{\GF }$-system $e \colon X \to \free{\GF }X$ and the parameter $Y = 0$.
\begin{equation}\label{diag:factfreeHQ}
\xymatrix{
    X \ar[rr]^{\quotsol {(\quotG_X \circ e)}} \ar[drr]_{\cansol e} && Q0 \\
    && \free{\GF }0 \ar[u]_{\quotG_0}
}
\end{equation}
Thus for $\free{\GF }$-systems the equality \eqref{eq:factQ} is immediate, because $\bbq{e} = \quotsol {(\quotG_X \circ e)}$ by Definition~\ref{def:quotSem} and $\quotsol {(\quotG_X \circ e)} = \quotG_0 \circ \cansol e = \quotG_0 \circ \bb{e}$ by commutativity of \eqref{diag:factfreeHQ}.

Starting instead from an $\GF $-system $e'$ based on state space $X$, with or without internal computations, consider the following chain of equalities:
\[ \bbq{e'} = \quotsol {(\quotG_X \circ \overline{e'})}  =\ \quotG_0 \circ \cansol{\overline{e'}} = \quotG_0 \circ \bb{e'}. \]
The first and third equalities are given by unfolding the definition of $\bbq{-}$ and $\bb{-}$, whereas the second one is due to commutativity of \eqref{diag:factfreeHQ} applied to the $\free{\GF }$-system $\overline{e'} \colon X \to \free{\GF }X$ in place of $e$. \qed
\end{proof}

\begin{corollary_for} {for:sound-quot}
\corSoundQuot
\end{corollary_for}
\begin{proof} The statement is immediately given by the following derivation
\[\bbq{e} = \quotG_0 \circ \bb{e} = \quotG_0 \circ \bb{\epselim e} = \bbq{\epselim e}\]
where the first and third equalities hold by Proposition \ref{prop:factQuotient} and the second equality by Theorem \ref{th:epselimsound}. \qed
\end{proof}

\section{Proof of Theorem~\ref{th:quotcomparewithHJS}}
Finally, we can prove Theorem~\ref{th:quotcomparewithHJS}. The following lemma provides sufficient conditions for lifting the quotient of an endofunctor to $\Kl(M)$.
 \begin{proposition}\label{lemma:quotient}
 Let $M,S\colon \C \to \C$ be monads such that there exists a distributive law $\lambda\colon SM \to MS$ and
 let $\lift{S}\colon \Kl(M) \to \Kl(M)$ be the corresponding lifting.
 Let $\quotG \colon S \To R$ be a monad quotient  such that
\begin{itemize}
\item[(c)] for each $X$, there is a map
 $\lambda_X'\colon RMX \to MRX$ making the following commute.
 $$\xymatrix{
 SMX \ar[d]_{\quotG_{MX}} \ar[r]^{\lambda_X} & MSX \ar[d]^{M{\quotG}_X}\\
 RMX \ar[r]_{\lambda'_X} & MRX
 }$$
 \end{itemize}
 Then $R$ lifts to a monad
 $\lift{R}\colon \Kl(M) \to \Kl(M)$ and $\lift{q} \colon \lift{S} \To \lift{R}$ defined as $\lift{\quotG_X}=\J(\quotG_X)$ is a monad quotient.
 \end{proposition}
\begin{proof}
 We first prove that $\lambda'\colon RM \To MR$ given by $\{\lambda'_X\}_{X}$ is a natural transformation.
 Let $f \colon X \to Y$ be a morphism in $\C$. As each $\quotG$-component is epi, it suffices to check that $MRf \circ \lambda'_X \circ \quotG_{MX} = \lambda'_Y \circ RMf \circ \quotG_{MX}$. For this purpose we construct the following cube.
 $$\xymatrix{
 RMX  \ar[rd]_{RMf} \ar[rr]^{\lambda'_X} && MRX \ar[rd]^{MRf}\\
 & RMY \ar[rr]^(0.7){\lambda'_Y} && MRY\\
 SMX \ar[rd]_{SMf} \ar[uu]^{\quotG_{MX}} \ar[rr]_(0.7){\lambda_X} && MSX \ar[rd]^{MSf} \ar[uu]|(0.3){M{\quotG}_X}\\
 & SMY \ar[uu]^(0.7){\quotG_{MY}} \ar[rr]_{\lambda_Y} && MSY \ar[uu]_{M{\quotG}_Y}
 }$$
The bottom face commutes by naturality of $\lambda$; the leftmost and the righmost faces commute by naturality of $\quotG$; the backward and the front face commute because of $(\dagger)$. It is therefore easy to see that $MRf \circ \lambda'_X \circ \quotG_{MX} = \lambda'_Y \circ RMf \circ \quotG_{MX}$.

Now, we prove that $\lambda'\colon RM \To MR$ is a distributive law of monads. The argument for the four diagrams is analogous, so we just show the one for $\eta_M$, depicted in the triangle $(1)$, below.
$$\xymatrix{
& RX \ar@{}[d]|{(1)} \ar[rd]^{R\eta^M_X} \ar[ld]_{\eta^M_{RX}} & & SX \ar@{}[ld]|{(2)} \ar[dd]^(0.7){\eta^M_{SX}} \ar[ll]_{\quotG_X} \ar[rd]^{S\eta^M_X}\\
MRX & & RMX \ar[ll]_{\lambda'_X} & & SMX \ar@{}[llu]|(0.3){(3)} \ar[ld]^{\lambda_X} \ar[ll]_(0.7){\quotG_{MX}}\\
& & & MSX \ar@{}[lu]|{(4)} \ar[ulll]^{M{\quotG}_X}
}$$
Observe that $(2)$ commutes by naturality of $\quotG$, $(3)$ commutes since $\lambda$ is a distributive law of monads and $(4)$ commute by $(\dagger)$.
Therefore the first equality of the following equation holds $$\lambda'_X \circ R\eta^M_X \circ \quotG_X = M{\quotG}_X \circ \eta^M_{SX}=\eta^M_{RX}\circ \quotG_X$$ and the second equality holds by naturality of $\eta^M$. The commutativity of $(1)$ follows since $\quotG_X$ is epi.

By Proposition~\ref{LiftProp}, and the fact that $\lambda'\colon RM \To MR$, then $R$ has a monad lifting $\lift{R} \colon \Kl(M) \to \Kl(M)$.

We now prove that $\lift{q} \colon \lift{S} \To \lift{R}$ is a monad morphism. First, we need to check that it is a natural transformation, that is for all morphisms $f \colon X \to Y$ in $\Kl(M)$, the following diagram commutes.
$$
\xymatrix{
\lift{S}X \ar[r]^{\J(\quotG_X)} \ar[d]_{\lift{S}f} & \lift{R}X \ar[d]^{\lift{R}f}\\
\lift{S}Y \ar[r]_{\J(\quotG_Y)} & \lift{R}Y
}
$$
By spelling out the definitions of $\J$ and $\lift{S}$, the above diagram corresponds to the following in $\C$.
$$
\xymatrix@C=2cm{
SX \ar@{}[rd]|{(1)} \ar[r]^{\eta_{SX}^M} \ar[d]_{Sf} & MSS \ar@{}[rd]|{(2)} \ar[r]^{M{\quotG}_X} \ar[d]|{MSf} & MRX \ar[d]^{MRf}\\
SMY\ar@{}[rd]|{(3)}  \ar[r]|{\eta^M_{SY}} \ar[d]_{\lambda_X} & MSMY \ar@{}[rd]|{(4)} \ar[r]|{M{\quotG}_{MY}} \ar[d]|{M\lambda_Y} & MRMY \ar[d]^{M\lambda'_Y}\\
MSY \ar[r]_{\eta^M_{MSY}} & MMSY \ar[r]_{MM{\quotG}_Y} & MMRY \ar[r]^{\mu^M_{RY}} & MRY
}
$$
Observe that $(1)$ and $(3)$ commute by naturality of $\eta^M$, $(2)$ commutes by naturality of $\quotG$ and $(4)$ commutes by $(\dagger)$.

Verifying that $\lift{q}$ is a also morphism of monads is immediate:
$\lift{q} \klafter \eta^{\lift{S}} = \J (q) \klafter \J (\eta^S) = \J (\eta^{R}) = \eta^{\lift{R}}$
and $\lift{q} \klafter \mu^{\lift{S}} = \J (q) \klafter \J (\mu^S) = \J (\mu^{R}) \klafter \J (Rq \circ \quotG_S)= \mu^{\lift{S}} \circ \lift{R}\lift{q}\circ \lift{\quotG_S}$.

All its components are epi since $\J$ is a left adjoint and thus preserves epis.
\qed
\end{proof}

\begin{theorem_for}{th:quotcomparewithHJS}
\thquotcomparewithHJS
\end{theorem_for}
\begin{proof} The conditions of Point~\ref{pt:QuotKleisli2} are guaranteed by Proposition~\ref{lemma:quotient}. In particular, the morphism $\lift{{\quot}} \colon \free{\lift{H}} \To \lift{\MM}$ is of the right type because $\lift{\free{H}} = \free{\lift{H}}$ by Proposition~\ref{prop:liftingfreemonad}. For point~\ref{pt:QuotKleisli3} we observe that, for $\Kl(M)$, $\lift{H}$, $\free{\lift{H}}$ and $\lift{H}\J X + \id$, proving Assumption~\ref{ass:quotient} amounts to show Assumption~\ref{ass:TheoryHSystems}, which we already did in Theorem~\ref{th:comparewithHJS}.

Thus it only remains to prove that for all $Y \in \Kl(M)$ an initial $\lift{\MM}(\id+Y)$-algebra exists. In virtue of Proposition~\ref{prop:liftingCoproduct}, it suffices to show that $\MM \: \Sets \to \Sets$ is accessible. The accessibility of the quotient $\MM$ of $\free{H}\colon \Sets \to \Sets$ is guaranteed from the fact that $\free{H}\: \Sets \to \Sets$ is accessible (Remark~\ref{rem:chain}\ref{pt:freeHaccessible}) and thus bounded (Remark~\ref{rem:chain}\ref{pt:bounded}) and that the quotients of bounded functors are also bounded.
\qed
\end{proof}

\section{Modeling Mazurkiewicz Trace Semantics}\label{app:mazurAss}

 The following statement allows to apply the framework of quotient semantics (Section \ref{ssec:quot}) to the modeling of Mazurkiewicz trace semantics (Section \ref{Sec:MazurTraces}). The functor $H$, the monads $\MM$ and $\free{H}$, the quotient of monads $\quot \: \free{H} \To \MM$ and the congruence relation $\equiv$ are as in Example \ref{Sec:MazurTraces}.

\begin{proposition}The monads $\Pow \: \Sets \to \Sets$ and $\MM \: \Sets \to \Sets$, the functor $H\: \Sets \to \Sets$ and the quotient of monads  $\quot \: \free{H} \To \MM$ satisfy the assumptions of Theorem \ref{th:quotcomparewithHJS}.
\end{proposition}
\begin{proof}
Clearly the functor $H \: \Sets \to \Sets$ is accessible. The remaining properties of $H$ and of the monad $\Pow \: \Sets \to \Sets$ are as in Theorem \ref{thm:HJS} and have been already verified in \cite{HasuoJS:07}. Thus it remains to show that the quotient $\quot \: \free{H} \To \MM$ satisfies condition $(c)$ of Theorem \ref{th:quotcomparewithHJS}. For this purpose, fix $X \in \Sets$. The desired morphism $\lambda'_X \: \MM\Pow X \to \Pow \MM X$ will be given by universal property of a standard coequalizer diagram induced by the congruence relation $\equiv\ \subseteq {A^* \times A^*}$. First we define the set $E_{\Pow X} \subseteq (\free{H}\Pow X \times \free{H}\Pow X)$ as
\[ E_{\Pow X} \ :=\ \{\big( (w, Y) (v,Y)\big) \mid w \equiv v\} \cup \{(w, v) \mid w \equiv v\}\]
Intuitively, $E_{\Pow X}$ is the set of equations on $\free{H}\Pow X$ induced by $\equiv$. There are evident projection maps $\pi_1,\pi_2 \: E_{\Pow X} \to \free{H} \Pow X$. It is immediate to verify that the following is a coequalizer diagram.
\[\xymatrix{
E_{\Pow X} \ar@<.5ex>[r]^{\pi_1} \ar@<-.5ex>[r]_{\pi_2} &
\free{H} \Pow X \ar[rr]^{{\quot}_{\Pow X}} &&
\MM \Pow X
}
\]
Also one can check that the morphism $\Pow {\quot}_X \circ \lambda_X \: \free{H} \Pow X \to \Pow \MM X$ (where $\lambda \: \free{H} \Pow \To \Pow \free{H}$ is a distributive law as in the statement of Theorem \ref{th:quotcomparewithHJS}) gives the same values if precomposed with $\pi_1$ or with $\pi_2$. Thus the universal property of coequalizer yields a unique morphism $\lambda'_X$ making the following commute.
\[\xymatrix@R=.4cm{
E_{\Pow X} \ar@<.5ex>[r]^{\pi_1} \ar@<-.5ex>[r]_{\pi_2} &
\free{H} \Pow X \ar[rr]^{{\quot}_{\Pow X}} \ar[dr]^{\lambda_X} &&
\MM \Pow X \ar@{-->}[dd]^{\lambda'_X} \\
&& \Pow \free{H} X \ar[dr]^{\Pow {\quot}_X}&& \\
& && \Pow \MM X
}
\]
Commutativity of the above diagram yields condition $(c)$ of Theorem \ref{th:quotcomparewithHJS}. \qed
\end{proof} 


\end{document}

\end{document}